\documentclass[10pt, conference, compsocconf]{IEEEtran}

\usepackage{amsfonts}

\newcommand{\set}[1]{\left\{ #1\right\}}

\newcommand{\sodass}{\,:\,}
\newcommand{\setGilt}[2]{\left\{ #1\sodass #2\right\}}

\newcommand{\realrange}[2]{\left[#1, #2\right]}

\newcommand{\unitrange}[2]{\realrange{0}{1}}

\newcommand{\llabel}[1]{\label{\labelprefix:#1}}
\newcommand{\labelprefix}{} 

\newcommand{\discussionsize}{\small}

\newenvironment{code}{\noindent
\begin{tabbing}%
\hspace{2em}\=\hspace{2em}\=\hspace{2em}\=\hspace{2em}\=\hspace{2em}\=%
\hspace{2em}\=\hspace{2em}\=\hspace{2em}\=\hspace{2em}\=\hspace{2em}\=%
\kill}{\end{tabbing}}

\newcommand{\labelcommand}{}
\newcommand{\captiontext}{}
\newsavebox{\codeparam}
\newcounter{lineNumber}
\newenvironment{disscodepos}[3]{%
\renewcommand{\labelcommand}{#2}%
\renewcommand{\captiontext}{#3}%
\sbox{\codeparam}{\parbox{\textwidth}{#3}}%
\begin{figure}[#1]\begin{center}\begin{code}\setcounter{lineNumber}{1}}{%
\end{code}\end{center}\caption{\llabel{\labelcommand}\captiontext}\end{figure}}

{\end{disscodepos}}

\newcommand{\If}       {{\bf if\ }}
\newcommand{\Is}       {:=}

\newdimen\endofsize\endofsize=0.5em
\def\endofbeweis{~\quad\hglue\hsize minus\hsize
                 \hbox{\vrule height \endofsize width
\endofsize}\par}

\usepackage[algo2e,ruled,vlined,linesnumbered, noend]{algorithm2e}
\newcommand{\ie}{i.e.}
\newcommand{\etal}{et~al.~}

\newcommand{\Id}[1]{\ensuremath{\text{{\sf #1}}}}
\newcommand{\E}{\mathrm{E}}

\SetKwFor{For}{for}{}{end for}%
\SetKwFor{PWhile}{while}{do in \hl{parallel}}{end while}
\SetKwProg{Function}{Function}{}{}
\SetKwFunction{Parallel}{Parallel}
\SetKwFunction{PerformMoves}{PerformMoves}
\SetKwFunction{ApplyMoves}{ApplyMoves}
\SetKwRepeat{Do}{do}{while}
\SetKw{Break}{break}
\SetKw{Continue}{continue}
\SetInd{0.5em}{0.5em}

\usepackage{amsmath}
\usepackage{amsthm}
\usepackage{amssymb}
\usepackage{color}
\usepackage{graphicx}
\usepackage[bookmarks=false]{hyperref}
\usepackage{soul}
\usepackage{balance}
\newcommand{\revision}[1]{#1}

\newtheorem{theorem}{Theorem}

\begin{document}
\title{High-Quality Shared-Memory Graph Partitioning}

\author{\IEEEauthorblockN{Yaroslav Akhremtsev}
\IEEEauthorblockA{Karlsruhe Institute of Technology (KIT)\\
yaroslav.akhremtsev@kit.edu}
\and
\IEEEauthorblockN{Peter Sanders}
\IEEEauthorblockA{Karlsruhe Institute of Technology (KIT)\\
peter.sanders@kit.edu}
\and
\IEEEauthorblockN{Christian Schulz}
\IEEEauthorblockA{Faculty of Computer Science \\University of Vienna \\
	christian.schulz@univie.ac.at
	}
}

\maketitle
\thispagestyle{plain}
\pagestyle{plain}
\begin{abstract}
Partitioning graphs into blocks of roughly equal size such that few
edges run between blocks is a frequently needed operation in
processing graphs. Recently, size, variety, and structural complexity of these networks has
grown dramatically. Unfortunately, previous approaches to parallel graph partitioning have problems
in this context since they often show a negative trade-off between speed and quality. We present an approach to multi-level shared-memory parallel graph partitioning that guarantees balanced solutions, shows high speed-ups for a variety of large graphs and yields very good quality independently of the number of cores used.
\revision{For example, on 31 cores, our algorithm partitions our largest test instance 
into 16 blocks cutting \emph{less than half} the number of edges than our main
competitor when both algorithms are given the same amount of time.}
Important ingredients include parallel label propagation for both coarsening and improvement,
parallel initial partitioning, a simple yet effective approach to parallel localized local search, and fast locality preserving hash tables.

\end{abstract}


\IEEEpeerreviewmaketitle

\section{Introduction}
Partitioning a graph into $k$ blocks of similar size such that few edges are cut
is a fundamental problem with many applications. For example, it often arises when processing a single graph on $k$ processors.


The graph partitioning problem is NP-hard 
and
  there is no approximation algorithm with a constant ratio factor
  for general graphs~\cite{bui1992finding}. Thus, to solve the graph partitioning problem
  in practice, one needs to use heuristics. A very common approach to partition a graph is the multi-level graph partitioning (MGP) approach. The main idea is to contract the graph in the \emph{coarsening} phase until it is small enough to be partitioned by
  more sophisticated but slower algorithms in the \emph{initial partitioning} phase. Afterwards, in the \emph{uncoarsening/local search} phase, the quality of the partition is improved on every level of the computed hierarchy
  using a local improvement algorithm.

  There is a need for shared-memory parallel graph partitioning algorithms that efficiently utilize all cores of a machine.
  This is due to the well-known fact that CPU technology increasingly provides more cores with relatively low clock rates in the
  last years since these are cheaper to produce and run.
  Moreover, shared-memory parallel algorithms implemented without message-passing libraries (e.g. MPI) usually give better
  speed-ups and running times than its MPI-based counterparts. 
  Shared-memory parallel graph partitioning algorithms can in turn also be used as a component of a distributed graph partitioner,
  which distributes parts of a graph to nodes of a compute cluster and then employs a shared-memory parallel
  graph partitioning algorithm to partition
  the corresponding part of the graph on a node level.

  \emph{Contribution:}
    We present a high-quality shared-memory parallel multi-level graph partitioning algorithm that parallelizes
    all of the three MGP phases -- coarsening, initial partitioning and refinement -- using \texttt{C++14} multi-threading.
    Our approach uses a parallel label propagation algorithm that is able
    to shrink large complex networks fast during coarsening.  Our
    parallelization of localized local search~\cite{kaffpa} is able to
    obtain high-quality~solutions and guarantee balanced partitions
    despite performing most of the work in mostly independent local searches
    of individual threads.  Using \emph{cache-aware hash tables} we
    limit memory consuption and improve locality.

  After presenting preliminaries and related work in Section~\ref{sec:prelim}, we explain details of the multi-level graph partitioning approach and
  the algorithms that we parallelize in Section~\ref{mgp}.
  Section~\ref{parallel_mgp} presents our approach to the parallelization of the multi-level graph partitioning phases.
  More precisely, we present a parallelization of label propagation with size-constraints~\cite{meyerhenke2014partitioning},
  as well as a  parallelization of $k$-way multi-try local search~\cite{kaffpa}.
  Section~\ref{impl_details} describes further optimizations.
  Extensive experiments are presented in Section~\ref{sec:experi}.
  Our approach scales comparatively better than other parallel partitioners
  and has considerably higher quality which does not degrade with increasing number of processors.

\section{Preliminaries}
\label{sec:prelim}
  \subsection{Basic concepts}
    %

\revision{Let $G=(V=\{0,\ldots, n-1\},E)$ be an undirected graph, where
$n = |V|$ and $m = |E|$.
We consider positive, real-valued edge and vertex weight functions $\omega$ and $c$ extending
them to sets, e.g., $\omega(M)\Is \sum_{x\in M}\omega(x)$. $N(v)\Is\setGilt{u}{\set{v,u}\in E}$
denotes the neighbors of $v$. The degree of a vertex $v$ is $d(v):=|N(v)|$.  $\Delta$ is the
maximum vertex degree. A vertex is a \emph{boundary vertex} if it is incident to a vertex in a different
block. We are looking for disjoint \emph{blocks} of vertices $V_1$,\ldots,$V_k$ that partition $V$; i.e.,
$V_1\cup\cdots\cup V_k=V$. The \emph{balancing constraint} demands that all blocks
have weight $c(V_i)\leq (1+\epsilon)\lceil\frac{c(V)}{k}\rceil=:L_{\max}$
for some imbalance parameter $\epsilon$.
We call a block $V_i$ \emph{overloaded} if its weight exceeds $L_{\max}$.
The objective is to minimize the total \emph{cut} $\omega(E\cap\bigcup_{i<j}V_i\times V_j)$.}
We define the gain of a vertex \revision{as the maximum decrease in cut size
when moving it to a different block.}
We denote the number of processing elements (PEs) as $p$.

A clustering is also a partition of the vertices. However, $k$ is usually not given in advance and the balance constraint is removed.
A size-constrained clustering constrains the size of the blocks of a clustering \revision{by a given upper bound $U$.}

An abstract view of the partitioned graph is a \emph{quotient graph}, in which vertices represent blocks and edges are induced by
connectivity between blocks.
The \emph{weighted} version of the quotient graph has vertex weights that are set to the weight of the corresponding block and edge weights
that are equal to the weight of the edges that run between the respective blocks.

In general, our input graphs $G$ have unit edge weights and vertex weights.
However, even those will be translated into weighted problems in the course of the multi-level algorithm.
In order to avoid a tedious notation, $G$ will denote the current state of the graph before and after a (un)contraction in the multi-level scheme
throughout this paper.

\revision{Atomic concurrent updates of memory cells are possible using the compare and swap operation
CAS($x$, $y$, $z$). If $x = y$ then this operation assigns $x \leftarrow z$ and returns $\textit{True}$; 
otherwise it returns $\textit{False}$.

We analyze algorithms using the concept of total \emph{work} (the
time needed by one processor) and \emph{span}; i.e., the time
needed using an unlimited number of processors~\cite{Blelloch96}.}
    
  \subsection{Related Work}
    There has been a \emph{huge} amount of research on graph partitioning so that we refer the reader 
    to~\cite{schloegel2000gph,GPOverviewBook,Walshaw07,buluc2016recent} for most of the material.
    Here, we focus on issues closely related to our main contributions.
    All general-purpose methods that are able to obtain good partitions for large real-world graphs are based on the multi-level principle.
    Well-known software packages based on this approach include Jostle~\cite{Walshaw07}, KaHIP~\cite{kaffpa},
    Metis~\cite{karypis1998fast} and Scotch~\cite{ptscotch}.

    Probably the fastest available distributed memory parallel code is the parallel version of Metis,
    ParMetis~\cite{karypis1996parallel}.
    This parallelization has problems maintaining the balance of the blocks since at
    any particular time, it is difficult to say how many vertices are assigned to a
    particular block. In addition, ParMetis only uses very simple greedy local search algorithms that do not yield high-quality solutions.
    Mt-Metis by LaSalle and Karypis~\cite{lasalle2016parallel, lasalle2013multi}
    is a shared-memory parallel version of the \texttt{ParMetis} graph partitioning framework.
    LaSalle and Karypis use a hill-climbing technique during refinement. The local search method is a simplification of $k$-way multi-try local search~\cite{kaffpa} in order to make it fast. 
    The idea is to find a set of vertices (hill) whose move to another block is beneficial
    and then to move this set accordingly. However, it is possible that several PEs move the same vertex. To handle this, each vertex
    is assigned a PE, which can move it exclusively. Other PEs use a
    message queue to send a request to move~this~vertex.


    PT-Scotch~\cite{ptscotch}, the parallel version of Scotch, is based on
    recursive bipartitioning. This is more difficult to parallelize than direct
    $k$-partitioning since in the initial bipartition, there is less parallelism
    available.  The unused processor power is used by performing several independent
    attempts in parallel. The involved communication effort is reduced by considering only vertices
    close to the boundary of the current partitioning (band-refinement).
    KaPPa~\cite{kappa} is a parallel matching-based MGP algorithm which is also restricted to the case
    where the number of blocks equals the number of processors used.
    PDiBaP~\cite{Meyerhenke12shapeCross} is a multi-level diffusion-based algorithm that is targeted at small- to
    medium-scale parallelism with dozens of processors.

    The label propagation clustering algorithm was initially proposed by Raghavan \etal\cite{labelpropagationclustering}.
    A single round of simple label propagation can be interpreted as the randomized agglomerative clustering approach proposed by Catalyurek and
    Aykanat~\cite{catalyurek1999hypergraph}.
    Moreover, the label propagation algorithm has been used to partition networks by Uganer and Backstrom~\cite{UganderB13}. The authors do not use
    a multi-level scheme and rely on a given or random partition  which is improved by combining the unconstrained label propagation approach with
    linear programming. The approach does not yield high quality~partitions.

    Meyerhenke \etal\cite{meyerhenke2017parallel} propose ParHIP, to partition large complex networks on distributed memory parallel machines.
    The partition problem is addressed by parallelizing and adapting the label propagation technique for graph coarsening and refinement. The
    resulting system is more scalable and achieves higher quality than
    the state-of-the-art systems like ParMetis~or~PT-Scotch.


\section{Multi-level Graph Partitioning}\label{mgp}
  We now give an in-depth description of the three main phases of a multi-level graph partitioning algorithm:
  coarsening, initial partitioning and uncoarsening/local search. In particular, we give a description of the sequential algorithms that we parallelize in the following sections. 
\revision{Our starting point here is the fast social configuration of KaHIP. For the development of the parallel algorithm, we add $k$-way multi-try local search scheme that gives higher quality, and improve it to perform less work than the original sequential version.}
  The original sequential implementations of these algorithms are contained in the KaHIP~\cite{kaffpa} graph partitioning framework. A general principle is to randomize tie-breaking whenever possible. This diversifies the search and allows improved solutions by repeated tries.

  \subsection{Coarsening}
To create a new level of a graph hierarchy, the rationale here is to compute a clustering with clusters that are bounded in size and then to
\emph{contract} each cluster into a supervertex. This coarsening procedure is repeated recursively until the coarsest graph is
small enough.
%
Contracting the clustering works by replacing each cluster with a single vertex.
The weight of this new vertex (or supervertex) is set to the sum of the weight of all vertices in the original cluster.
There is an edge between two vertices $u$ and $v$ in the contracted graph if the
two corresponding clusters in the clustering are adjacent to each other in $G$;
\ie, if the cluster of $u$ and the cluster of $v$ are connected by at least one edge. The weight of an edge $(A, B)$
is set to the sum of the weight of edges that run between cluster~$A$ and cluster~$B$ of the clustering.
The hierarchy created in this recursive manner is then used by the partitioner. Due to the way the contraction is defined,
it is ensured that a partition of the coarse graph corresponds to a partition of the finer graph with the same
cut and balance.
We now describe the clustering algorithm that we parallelize.

    \subsubsection*{Clustering} \label{coarsening:clustering}
      We denote the set of all clusters as $C$ and the cluster ID of a vertex $v$ as $C[v]$. There are a variety of
      clustering algorithms. Some of them build clusters of size two (matching algorithms) and
      other build clusters with size less than a given upper bound.
      In our framework, we use the label propagation algorithm by Meyerhenke et al.~\cite{meyerhenke2014partitioning} that
      creates a clustering fulfilling a size-constraint.

      The size constrained label propagation algorithm works in iterations; \ie, the algorithm is repeated $\ell$ times,
      where $\ell$ is a tuning parameter. Initially, each vertex is in its own cluster ($C[v] = v$) and all vertices are
      put into a queue $Q$ in increasing order of their degrees.
      During each iteration, the algorithm iterates over all vertices in $Q$.
      A neighboring cluster $\mathcal{C}$ of a vertex $v$ is called \emph{eligible} if $\mathcal{C}$ will not
      become overloaded once $v$ is moved to $\mathcal{C}$. When a vertex $v$ is visited, it is \emph{moved} to
      the eligible cluster that has the strongest connection to~$v$; \ie, it is moved
      to the eligible cluster $\mathcal{C}$ that maximizes $\omega(\{(v, u) \mid u \in N(v) \cap \mathcal{C} \})$.
      If a vertex changes its
      cluster ID then all its neighbors are added to a queue $Q'$ for the next iteration.
      At the end of an iteration, $Q$ and $Q'$ are swapped, and the algorithm proceeds with the next iteration.
      The sequential running time of one iteration of the algorithm is $\mathcal{O}(m + n)$.

      \subsection{Initial Partitioning}
      We adopt the algorithm from KaHIP~\cite{kaffpa}: After
      coarsening, the coarsest level of the hierarchy is partitioned
      into $k$ blocks using a recursive bisection
      algorithm~\cite{kernighansome}.  More
      precisely, it is partitioned into two blocks and then the
      subgraphs induced by these two blocks are recursively
      partitioned into $\lceil \frac{k}{2} \rceil$~and~$\lfloor
      \frac{k}{2} \rfloor$ blocks each. Subsequently, this partition
      is improved using local search and flow techniques.  To get a
      better solution, the coarsest graph is partitioned into $k$
      blocks $I$ times and the best solution is returned.

\subsection{Uncoarsening/Local Search} \label{seq_ls}
    After initial partitioning, a local search algorithm is applied to
    improve the cut of the partition.  When local search has finished,
    the partition is transferred to the next finer graph in the
    hierarchy; \ie, a vertex in the finer graph is assigned the block
    of its coarse representative.  This process is then repeated for
    each level of the hierarchy.

    There are a variety of local search algorithms: size-constraint label propagation,
    Fiduccia-Mattheyses $k$-way local search~\cite{fiduccia1982lth},
    max-flow min-cut based local search~\cite{kaffpa},
    $k$-way multi-try local search~\cite{kaffpa} \ldots .
    Sequential versions of KaHIP use combinations of those. 
    Since $k$-way local search is P-complete~\cite{savage1991parallelism},
    our algorithm uses size-constraint label propagation in combination with $k$-way multi-try local
    search. More precisely, the size-constraint label propagation algorithm can be used as a fast local search algorithm if one starts from a partition of the graph instead of a clustering and uses the size-constraint of the partitioning problem.
    On the other hand, $k$-way multi-try local search is able to find high quality solutions. Overall, this combination allows us to achieve a parallelization with good solution quality and~good~parallelism.

    We now describe multi-try $k$-way local search (MLS).  In contrast
    to previous $k$-way local search methods MLS is not initialized
    with \emph{all} boundary vertices; that is, not all boundary vertices
    are eligible for movement at the beginning.  Instead, the method
    is repeatedly initialized with a \emph{single} boundary vertex.
    This enables more diversification and has a better chance of
    finding nontrivial improvements that begin with negative gain
    moves~\cite{kaffpa}.

    The algorithm is organized in a nested loop of global and local iterations.
    A global iteration works as follows.
    Instead of putting \emph{all} boundary vertices directly into a priority queue,
    boundary vertices under consideration are put into a todo list $T$.
    Initially, all vertices are unmarked.
    Afterwards, the algorithm repeatedly chooses and removes
    a random vertex $v \in T$. If the vertex is unmarked, it starts to perform $k$-way local search around $v$,
    marking every vertex that is moved during this search.
    More precisely, the algorithm inserts $v$ and $N(v)$ into a priority queue using gain values as
    keys and marks them. Next, it extracts a vertex with a maximum key from the priority queue and
    performs the corresponding move.
    If a neighbor of the vertex is unmarked then it is marked and inserted
    in the priority queue.
    If a neighbor of the vertex is already in the priority queue then its key (gain) is updated.
    Note that not every move can be performed due to the size-constraint on the blocks.
    The algorithm stops when the adaptive stopping rule by Osipov and Sanders~\cite{osipov2010n} decides to stop
    or when the priority queue is empty.
    More precisely, if the overall gain is negative then the stopping rule 
    estimates the probability that the overall gain will become positive again and signals to stop if this is unlikely.
    In the end, the best partition that has been seen during the process is reconstructed.
    In one local iteration, this is repeated until the todo list is empty.
    \revision{
   	After a local iteration, the algorithm reinserts moved vertices into the todo list (in randomly shuffled order).
        If the achieved gain improvement is larger than a certain percentage (currently $10$ \%) of the total improvement during the current global iteration,
        it continues to perform moves around the vertices currently in the todo list (next local iteration).
        This allows to further decrease the cut size without significant impact to the running time. 
        When improvements fall below this threshold, another (global) iteration is started that
        initializes the todo list with all boundary vertices.
        After a fixed number of global iterations (currently~3), the MLS algorithm stops.
        Our experiments show that three global iterations is a fair trade-off between the running time and
        the quality of the partition.
        This nested loop of local and global iterations is an improvement over the original MLS search from~\cite{kaffpa}
        since they allow for a better control of the running time of the algorithm.
	}

    The running time of one local iteration
    is $\mathcal{O}(n + \sum_{v \in V} d(v)^2)$.
    Because each vertex can be moved only once during a local iteration
    and we update the gains of its neighbors using a bucket heap.
    Since we update the gain of a vertex at most $d(v)$ times,
    the $d(v)^2$ term is the total cost to update the gain of a vertex $v$.
    Note, that this is an upper bound for the worst case, usually local
    search stops significantly earlier due the stopping rule or an empty priority queue.

\section{Parallel Multi-level Graph Partitioning} \label{parallel_mgp}
\revision{Profiling the sequential algorithm shows that each of the components of the multi-level scheme has a significant contribution to the overall algorithm.}
Hence, we now describe the parallelization of each phase of the multi-level algorithm described above.
The section is organized along the phases of the multi-level scheme: first we show how to parallelize coarsening, then initial partitioning and finally uncoarsening.
\revision{Our general approach is to avoid bottlenecks as well as performing independent work as much as possible.}

\subsection{Coarsening}
In this section, we present the parallel version of the size-constraint label
propagation algorithm to build a clustering and
the parallel contraction algorithm.

\subsubsection*{Parallel Size-Constraint Label Propagation} \label{parallel_mgp:clustering}
\revision{To parallelize the size-constraint label propagation algorithm,
we adapt a clustering technique by Staudt and Meyerhenke~\cite{staudt2016engineering} to coarsening.
Initially, we sort the vertices by increasing degree using the fast
parallel sorting algorithm by Axtmann et al.~\cite{DBLP:conf/esa/AxtmannWF017}.
We then form work packets representing a roughly equal amount of work and insert them into a TBB
concurrent queue~\cite{TBB}.
Note that we also tried the work-stealing approach from~\cite{singler2007mcstl} but it showed worse running times.
Our constraint is that a packet can contain at most vertices with a total number of $B$ neighbors.
We set $B=\max(1\,000,\sqrt{m})$ in our experiments -- the $1\,000$ limits contention for small instances
and the term $\sqrt{m}$ further reduces contention for large instances.
Additionally, we have an empty queue
$Q'$ that stores packets of vertices for the next iteration.
During an
iteration, each PE checks if the queue $Q$ is not empty, and if so it
extracts a packet of active vertices from the queue.  A PE then chooses
a new cluster for each vertex in the currently processed packet.  A
vertex is then moved if the cluster size is still feasible to take on
the weight of the vertex. Cluster sizes are updated atomically using a CAS instruction. This is important to guarantee that the size constraint is not violated.
Neighbors of moved vertices are inserted into a packet for the next iteration. If the sum of vertex degrees in that packet
exceeds the work bound $B$ then this packet is inserted into 
queue~$Q'$ and a new packet is created for subsequent vertices.}
When the queue~$Q$ is empty, the main PE exchanges~$Q$ and~$Q'$ and we proceed with the next iteration.
One iteration of the algorithm can be done with $\mathcal{O}(n + m)$ work and
$\mathcal{O}(\frac{n + m}{p} + \log p)$ span.

\subsubsection*{Parallel Contraction}\label{parallel_mgp:contraction}
The contraction algorithm takes a graph $G = (V, E)$ as well as a clustering $C$ and constructs
a coarse graph $G'=(V', E')$.
The contraction process consists of three phases: the remapping of cluster IDs to a consecutive set of IDs,
edge weight accumulation, and the construction of the coarse graph.
The remapping of cluster IDs assigns new IDs in the range $[0, |V'| - 1]$ to the clusters where $|V'|$ is the 
number of clusters in the given clustering. 
We do this by
calculating a prefix sum on an
array that contains ones in the positions equal to the current cluster IDs. This
phase runs in $\mathcal{O}(n)$ time when it is done sequentially.
Sequentially, the edge weight accumulation step calculates weights of edges in $E'$ using hashing.
More precisely, for each cut edge $(v, u) \in E$ we insert a pair $(C[v], C[u])$ such that $C[v] \not= C[u]$ into a hash table and
accumulate weights for the pair if it is already contained in the table. Due to hashing cut edges, the expected 
running time of this phase is $\mathcal{O}(|E'| + m)$.
To construct the coarse graph we iterate over all edges $E'$ contained in the hash table.
This takes time $\mathcal{O}(|V'| + |E'|)$.
Hence, the total expected running time to compute the coarse graph is $\mathcal{O}(m + n + |E'|)$ when~run~sequentially.


The parallel contraction algorithm works as follows.
First, we remap the cluster IDs using the parallel prefix sum algorithm by Singler et al.~\cite{singler2007mcstl}.
Edge weights are accumulated by iterating over the edges of the original graph in parallel.
We use the concurrent hash table of Maier and Sanders~\cite{maier2016concurrent} initializing it
with a capacity of $\min(\mathrm{avg\_deg} \cdot |V'|, |E'| / 10)$.
Here $\mathrm{avg\_deg} = 2|E|/|V|$ is the average degree of $G$ since we hope that the average degree of $G'$ remains the same.
The third phase is performed sequentially in the current implementation since profiling indicates that it is so fast
that it is not a bottleneck.

\subsection{Initial Partitioning}\label{parallel_mgp:init_part}
To improve the quality of the resulting partitioning of the coarsest graph $G' = (V', E')$,
we partition it into $k$ blocks $\max(p, I)$ times instead of $I$ times.
We perform each partitioning step independently in parallel using different random seeds.
To do so, each PE creates a copy of the coarsest graph and runs KaHIP sequentially on it. 
Assume that one partitioning can be done in $T$ time. Then
$\max(p, I)$ partitions can be built
with $\mathcal{O}(\max(p, I) \cdot T + p\cdot (|E'| + |V'|))$ work
and $\mathcal{O}(\frac{\max(p, I) \cdot T}{p} + |E'| + |V'|)$ span,
where
the additional terms $|V'|$ and $|E'|$ account for the time each PE copies the~coarsest~graph.

\subsection{Uncoarsening/Local Search} \label{parallel_mgp:refinement}
Our parallel algorithm first uses size-constraint parallel label propagation to improve the current partition and afterwards applies our parallel MLS. \revision{The rationale behind this combination is that label propagation is fast and easy to parallelize and will do all the easy improvements. Subsequent MLS will then invest considerable work to find a few nontrivial improvements.
  In this combination, only few nodes actually need be moved globally which makes it easier to
  parallelize MLS scalably.}
When using the label propagation algorithm to improve a partition, we set the upper bound $U$ to the size-constraint of the partitioning problem $L_{\max}$.

\revision{Parallel MLS works in a nested loop of local and global iterations as in the sequential version.
Initialization of a global iteration uses a simplified parallel shuffling algorithm
where each PE shuffles the nodes it considers into a local bucket and then the queue is
made up of these buckets in random order.}
During a local iteration, each PE extracts vertices
from a producer queue $Q$.
Afterwards, it performs \textit{local} moves around it;
that is, global block IDs and the sizes of the blocks remain \textit{unchanged}.
When the producer queue~$Q$ is empty, the algorithm applies the best found sequences of moves to the global data structures.
Pseudocode of one global iteration of the algorithm can be found in Algorithm~\ref{algorithm:parallel_multitry}.
In the paragraphs that follow, we describe
how to perform local moves in \texttt{PerformMoves} 
and then how to update the global data structures in \texttt{ApplyMoves}.

\begin{algorithm2e}[t]
	\label{algorithm:parallel_multitry}
	\normalsize
	\KwIn{Graph $G=(V, E)$; queue $Q$; \revision{threshold~$\alpha < 1$}}
        \tcp{all vertices not moved}
	\PWhile{$Q$ is \textbf{not} empty} {\label{alg:start}

		$v = Q.pop()$\;

		\lIf{$v$ is moved}{\Continue}

        \BlankLine
		$V_{pq} \gets v \cup \{w \in N(v): w\text{ is not moved}\}$\;

        \BlankLine
		\tcp{priority queue with gain as key}
		$PQ \gets \{(gain(w), w) : w \in V_{pq}\}$\;
        \BlankLine
		\tcp{try to move boundary vertices}
		\PerformMoves{G, PQ}\;
	}
        
	
	$stop \gets \mathit{true}$ \tcp*{signal other PEs to stop}
	\If{main thread}{
		\revision{
		$gain \gets \ApplyMoves{G, Q}$ \\
		\If{$gain > \alpha \cdot total\_gain$}{
			$total\_gain \gets total\_gain + gain$; Go to \ref{alg:start}\;
		}
		}
	}
	\caption{Parallel Multi-try $k$-way Local Search.}
\end{algorithm2e}


\subsubsection*{Performing moves (\texttt{PerformMoves})}
Starting from a single boundary vertex, each PE moves vertices to find a sequence of moves that decreases the cut.
However, all moves are local; that is, they do not affect the current global partition --
moves are stored in the local memory of the PE performing them.
To perform a move, a PE chooses a vertex with maximum gain and marks it so that
other PEs cannot move it. Then, we update the sizes of the affected blocks and save the move.
During the course of the algorithm, we store the sequence of moves yielding the best cut.
We stop if there are no moves to perform or the adaptive stopping rule signals the algorithm to stop.
When a PE finished, the sequences of moves yielding the largest decrease in the edge cut is returned.

\subsubsection* {Implementation Details of \texttt{PerformMoves}} \label{impl_details:perform_moves}
In order to improve scalability, only the array for marking moved
vertices is global.  Note that within a local iteration, only bits in this
array are set (using CAS) and they are never unset.  Hence, the
marking operation can be seen as priority update operation (see Shun
et al.~\cite{shun2013reducing}) and thus causes only little
contention.  The algorithm keeps a local array of block sizes, a
local priority queue, and a local hash table storing changed block IDs
of vertices.  Note that since the local hash table is small, it often
fits into cache which is crucial for parallelization due to memory
bandwidth limits.  When the call of \texttt{PerformMoves} finishes and
the thread executing it notices that the queue $Q$ is empty, it sets a
global variable to signal the other threads to finish the current
call of the function \texttt{PerformMoves}.

Let each PE process a set of edges $\mathcal{E}$ and a set of vertices $\mathcal{V}$.
  Since each vertex can be moved only by one PE and moving a vertex requires the gain computation of its neighbors, the span of the function
  \texttt{PerformMoves} is $\mathcal{O}(\sum_{v \in \mathcal{V}} \sum_{u \in N(v)}d(u) + |\mathcal{V}|) =
  \mathcal{O}(\sum_{v \in \mathcal{V}} d^2(v) + |\mathcal{V}|)$ since the gain of a vertex~$v$ can be updated
  at most $d(v)$ times.
  Note that this is a pessimistic bound and it is possible to implement this function with
  $\mathcal{O}(|\mathcal{E}| \log \Delta + |\mathcal{V}|)$ span.
  In our experiments, we use the implementation with the former running time since it requires less memory 
  and the worst case -- the gain of a vertex $v$ is updated $d(v)$ times -- is quite unlikely.

\subsubsection*{Applying Moves (\texttt{ApplyMoves})}\label{paralle_mgp:apply_moves}
Let $M_i = \{B_{i1}, \dots\}$ denote the set of sequences of moves performed by PE~$i$,
where $B_{ij}$ is a set of moves performed by $j$-th call of \texttt{PerformMoves}.
We apply moves sequentially in the following order $M_1, M_2, \dots, M_p$.
We can not apply the moves directly in parallel since a move done by one PE can affect a move done by another PE.
More precisely, assume that we want to move a vertex $v \in B_{ij}$ but we have already moved its neighbor $w$ on a different PE.
Since the PE only knows local changes, it calculates the gain to move $v$ (in \texttt{PerformMoves}) according 
to the old block ID of $w$.
If we then apply the rest of the moves in $B_{ij}$ it may even increase the cut.
To prevent this, we recalculate the gain of each move in a given sequence and remember the best cut.
If there are no affected moves, we apply all moves from the sequence. Otherwise we apply only the part of 
the moves that gives the best cut with respect to the correct gain values.
\revision{
Finally, we insert all moved vertices into the queue $Q$.
}
Let~$M$ be the set of all moved vertices during this procedure. The overall running time is then given by 
$\mathcal{O}(\sum_{v \in M} d(v))$.
\revision{Note that our initial partitioning algorithm generates balanced solutions. Since moves are applied sequentially our parallel local search algorithm maintains balanced solutions; \ie~the balance constraint of our solution is never violated.}

\subsection{Differences to Mt-Metis}
    \revision{We now discuss the differences between our algorithm and Mt-Metis.
    In the coarsening phase, our approach uses a cluster contraction scheme while Metis is using a matching-based scheme.
Our approach is especially well suited for networks that have a pronounced and hierarchical cluster structure. For example, in networks that contain star-like structures, a matching-based algorithm for coarsening matches only a single edge within these structures and hence cannot shrink the graph effectively. Moreover, it may contract
``wrong'' edges such as bridges. 
    Using a clustering-based scheme, however, ensures that the graph sizes shrink very fast in the multi-level
    scheme~\cite{meyerhenke2017parallel}.  
    The general initial partitioning scheme is similar in both algorithms.
    However, the employed sequential techniques differ because different sequential tools (KaHIP and Metis) are used to partition the coarsest graphs.
    }
    In terms of local search, unlike Mt-Metis, our approach guarantees that the updated partition is
    balanced if the input partition is balanced and that the cut can only decrease or stay the same.
    The hill-climbing technique, however, may increase the cut of the input partition or may compute an
    imbalanced partition even if the input partition is balanced.
    Our algorithm has these guarantees since each PE performs moves of vertices locally in parallel.
    When all PEs finish, one PE globally applies the best sequences of local moves computed by all PEs.
    Usually, the number of applied moves is significantly smaller than the number
    of the local moves performed by all PEs, especially on large graphs.
    Thus, the main work is still made in parallel.
    Additionally, we introduce a cache-aware hash table in the following section
    that we use to store local changes
    of block IDs made by each PE. This hash table is more compact than an array and
    takes the locality of data~into~account.
    
\section{Further Optimization} \label{impl_details}
In this section, we describe further optimization techniques that we use to achieve better speed-ups
and overall speed.
More precisely, we use cache-aligned arrays to mitigate the problem of false-sharing, the TBB scalable
allocator~\cite{TBB} for concurrent memory allocations and pin threads to cores to
avoid rescheduling overheads. Additionally, we use a cache-aware hash table which we describe now.
In contrast to usual hash tables, this hash table allows us to exploit locality of data and hence to reduce the overall running time of the algorithm.

\subsection{Cache-Aware Hash Table}
  The main goal here  is to improve the performance of our algorithm on large graphs.
  For large graphs, the gain computation in the MLS routine takes most of the time.
  Recall, that computing the gain of a vertex requires a local hash table.
  Hence, using a cache-aware technique reduces the overall running time.
  A cache-aware hash table combines both properties of an array and a hash table.
  It tries to store data with similar integer keys within the same cache line,
  thus reducing the cost of subsequent accesses to these keys.
  On the other hand, it still consumes less memory than an array which
  is crucial for the hash table to fit into caches.

  We implement a cache-aware hash table using the linear probing technique and the tabular hash
  function~\cite{Patrascu:2011:PST:1993636.1993638}.
  Linear probing typically outperforms other collision resolution techniques in practice
  and the computation of the tabular hash function can be done with a small overhead.
  The tabular hash function works as follows.
  Let $x = x_1 \dots x_k$ be a key to be hashed, where $x_i$ are $t$ bits of the binary representation of~$x$.
  Let $T_i, i \in [1, k]$ be tables of size $2^t$, where each element is a random $32$-bit integer.
 Using $\oplus$ as exclusive-or operation, the tabular hash function is then defined as follows:\[h(x) = T_1[x_1] \oplus \dots \oplus T_k[x_k].\]

\subsubsection*{Exploiting Locality of Data}
  As our experiments show, the distribution of keys that we access during the computation of the gains is not uniform.
  Instead, it is likely that the time between accesses to two consecutive keys is small.
  On typical systems currently used, the size of a cache line is 64 bytes (16 elements with 4 bytes each).
  Now suppose our algorithm accesses 16 consecutive vertices one after another.
  If we would use an array storing the block IDs of all vertices instead of a hash table, we can access all block IDs of the vertices with only one cache miss.
  A hash table on the other hand does not give any locality guarantees.
  On the contrary, it is very probable that consecutive keys are hashed to completely different parts of the hash table.
  However, due to memory constraints we can not use an array to store block IDs for each PE in the \texttt{PerformMoves} procedure.

 However, even if the arrays fit into memory this would be problematic. To see this let $|L2|$ and $|L3|$ be the sizes of L2 and L3 caches of a given system, respectively, and let $p'$ be the number of PEs used per a socket.
For large graphs, the array may not fit into $\max(|L2|, \frac{|L3|}{p'})$ memory.
  In this case, each PE will access its own array in main memory which affects the running time due to the available memory bandwidth.
  Thus, we want a compact data structure that fits into $\max(|L2|, \frac{|L3|}{p'})$ memory most of the time \emph{and} preserve the
  locality guarantees of an array.

  For this, we modify the tabular hash function from above according to Mehlhorn and Sanders~\cite{MehSan08}.
  More precisely, let $x = x_1 \dots x_{k - 1} x_k$, where $x_k$ are the $t'$ least
  significant bits of $x$ and $x_1, \dots, x_{k - 1}$ are $t$ bits each.
  Then we compute the tabular hash function as follows:
  \[h(x) = T_1[x_1] \oplus \dots \oplus T_{k - 1}[x_{k - 1}] \oplus x_k.\]
  This guarantees that if two keys $x$ and $x'$ differ only in first $t'$ bits and,
  hence, $|x - x'| < 2^{t'}$ then~$|h(x) - h(x')| < 2^{t'}$.
  Thus, if $t' = \mathcal{O}(\log c)$, where $c$ is the size of a cache 
  line, then $x$ and $x'$ are in the same cache line when accessed. This hash function
  introduces at most $2^{t'}$ additional collisions since if we do not consider~$t'$
  least significant bits of a key then at most $2^{t'}$ keys have the same remaining bits.
  In our experiments, we choose $k = 3, t' = 5, t = 10$.

\section{Experiments}
\label{sec:experi}

  \subsection{Methodology}
    We implemented our algorithm \texttt{Mt-KaHIP} (\revision{Multi-threaded Karlsruhe High Quality Partitioning})  within the \texttt{KaHIP}~\cite{kaffpa} framework using \texttt{C++} 
 and the \texttt{C++14} multi-threading library. We plan to make our program available in the framework.
    All binaries are built using \texttt{g++-5.2.0} with the \texttt{-O3} flag and 64-bit index data types.
    We run our experiments on two machines.
    \revision{
    Machine $A$ is an Intel Xeon~E5-2683v2 (2 sockets, 16 cores with Hyper-Threading, 64 threads)
    running at 2.1 GHz with 512GB RAM.
    Machine $B$ is an Intel Xeon~E5-2650v2 (2 sockets, 8 cores with Hyper-Threading, 32 threads)
    running at 2.6 GHz with 128GB RAM.}

    We compare ourselves to \textsf{Mt-Metis 0.6.0} using the default configuration with hill-climbing being enabled (\textit{Mt-Metis}) as well as 
    sequential \textsf{KaHIP 2.0} using the \texttt{fast social} configuration (\textit{KaHIP}) and
    \textsf{ParHIP 2.0}~\cite{meyerhenke2017parallel} using the \texttt{fast social} configuration~(\textit{ParHIP}).
	According to LaSalle and Karypis~\cite{lasalle2013multi} 
    \texttt{Mt-Metis} has better speed-ups and running times compare to~\texttt{ParMetis} and~\texttt{Pt-Scotch}.
    At the same time, it has similar quality of the partition.
    Hence, we do not perform experiments with~\texttt{ParMetis} 
    and~\texttt{Pt-Scotch}.

    Our default value of allowed imbalance is 3\% -- this is one of the values used in~\cite{walshaw2000mpm}.
    \revision{We call a solution imbalanced if at least one block exceeds this amount.}
    By default, we perform ten repetitions for every algorithm using different random seeds for initialization
    and report the arithmetic average of computed cut size and running time on a per instance (graph and number of blocks $k$) basis.
	When further averaging over multiple instances, we use the \emph{geometric mean} for quality and time per edge quantities
	and the \emph{harmonic mean} for the relative speed-up
	in order to give every instance a comparable influence on the final score.  
    If at least one repetition returns an imbalanced partition of an instance then
    we mark this instance imbalanced.
	Our experiments focus on the cases $k \in \{16, 64\}$  and $p \in \{1, 16, 31\}$ to save running time and to
    keep the experimental evaluation simple.

    We use performance plots to present the quality comparisons and scatter plots to present the speed-up and the running time comparisons.
    A curve in a performance plot for algorithm~X is obtained as follows:
    For each instance (graph and $k$), we calculate the normalized value $1 - \frac{\texttt{best}}{\texttt{cut}}$,
    where \texttt{best} is the best cut obtained by any of the considered algorithms and \texttt{cut} is the cut of algorithm~X.
    These values are then sorted. Thus, the result of the best algorithm is in the bottom of the plot.
	We set the value for the instance above $1$  if an algorithm builds an imbalanced partition.
    Hence, it is in the top of~the~plot.
    

    \subsubsection*{Algorithm Configuration}
      Any multi-level algorithm has a considerable number of choices between algorithmic components and tuning parameters.
      We adopt parameters from the coarsening and initial partitioning phases of \texttt{KaHIP}.
      The \texttt{Mt-KaHIP} configuration uses $10$ and $25$ label propagation iterations during coarsening and refinement, respectively,
      partitions a coarse graph $\max(p, 4)$ times in initial partitioning and uses three global iterations of parallel MLS in the refinement~phase.
    \begin{table}[t]
      \centering
      \begin{tabular}{l|r|r||r||r}
      graph & $n$ & $m$ & type & ref. \\
              \hline
              \hline
      amazon                                      & $\approx$0.4M             & $\approx$2.3M        &C& \cite{snap}\\
      youtube                                     & $\approx$1.1M             & $\approx$3.0M        &C& \cite{snap}\\
      amazon-2008                                 & $\approx$0.7M             & $\approx$3.5M        &C& \cite{webgraphWS}\\
      in-2004                                     & $\approx$1.4M             & $\approx$13.6M       &C& \cite{webgraphWS}\\
      eu-2005                                     & $\approx$0.9M             & $\approx$16.1M       &C& \cite{webgraphWS}\\
      packing                                     & $\approx$2.1M             & $\approx$17.5M       &M& \cite{benchmarksfornetworksanalysis}\\
      del23                                       & $\approx$8.4M             & $\approx$25.2M       &M& \cite{kappa}\\
      hugebubbles-00                              & $\approx$18.3M            & $\approx$27.5M       &M& \cite{benchmarksfornetworksanalysis}\\
      channel                                     & $\approx$4.8M             & $\approx$42.7M       &M& \cite{benchmarksfornetworksanalysis}\\
      cage15                                      & $\approx$5.2M             & $\approx$47.0M       &M& \cite{benchmarksfornetworksanalysis}\\
      europe.osm                                  & $\approx$50.9M            & $\approx$54.1M       &M& \cite{BMSW13}\\
      enwiki-2013                                 & $\approx$4.2M             & $\approx$91.9M       &C& \cite{webgraphWS}\\
      er-fact1.5-scale23                          & $\approx$8.4M             & $\approx$100.3M      &C& \cite{BMSW13}\\
      hollywood-2011                              & $\approx$2.2M             & $\approx$114.5M      &C& \cite{webgraphWS}\\
      rgg24                                       & $\approx$16.8M            & $\approx$132.6M      &M& \cite{kappa}\\
      rhg                                         & $\approx$10.0M            & $\approx$199.6M      &C& \cite{DBLP:conf/isaac/LoozMP15}\\
      del26                                       & $\approx$67.1M            & $\approx$201.3M      &M& \cite{kappa}\\
      uk-2002                                     & $\approx$18.5M            & $\approx$261.8M      &C& \cite{webgraphWS}\\
      nlpkkt240                                   & $\approx$28.0M            & $\approx$373.2M      &M& \cite{UFsparsematrixcollection}\\
      arabic-2005                                 & $\approx$22.7M            & $\approx$553.9M      &C& \cite{webgraphWS}\\
      rgg26                                       & $\approx$67.1M            & $\approx$574.6M      &M& \cite{kappa}\\
      uk-2005                                     & $\approx$39.5M            & $\approx$783.0M      &C& \cite{webgraphWS}\\
      webbase-2001                                & $\approx$118.1M           & $\approx$854.8M      &C& \cite{webgraphWS}\\
      it-2004                                     & $\approx$41.3M            & $\approx$1.0G        &C& \cite{webgraphWS}\\
      \end{tabular}
      \caption{Basic properties of the benchmark set with a rough type classification. C stands for complex networks, M is used for mesh type networks.}
       \label{tab:scalefreegraphstable}
       \vspace{-0.75cm}
    \end{table}
    \subsubsection*{Instances}
      We evaluate our algorithms on a number of large graphs. These graphs are collected
      from~\cite{benchmarksfornetworksanalysis,UFsparsematrixcollection,BoVWFI,snap,DBLP:conf/isaac/LoozMP15,BMSW13}.
      Table~\ref{tab:scalefreegraphstable} summarizes the main properties of the benchmark set.
      Our benchmark set includes a number of graphs from numeric simulations as well as complex networks (for the latter
      with a focus on social networks and web graphs).

      The \Id{rhg} graph is a complex network generated with NetworKit~\cite{DBLP:conf/isaac/LoozMP15} according to the
      \emph{random hyperbolic graph} model~\cite{Krioukov2010}.
      In this model vertices are represented as points in the hyperbolic
      plane; vertices are connected by an edge if their hyperbolic distance is below a threshold.
      Moreover, we use the two graph families \Id{rgg} and \Id{del} for
      comparisons.
      \Id{rgg$X$} is a \emph{random geometric graph} with
      $2^{X}$ vertices where vertices represent random points in the (Euclidean) unit square and edges
      connect vertices whose Euclidean distance is below $0.55 \sqrt{ \ln n / n }$.
      This threshold was chosen in order to ensure that the graph is almost certainly connected.
      \Id{del$X$} is a Delaunay triangulation of $2^{X}$
      random points in the unit square.
      The graph \Id{er-fact1.5-scale23} is generated using the Erd{\"o}s-R{\'e}nyi $G(n, p)$ model with $p = 1.5 \ln n / n$.

\subsection{Quality Comparison}

  In this section, we compare our algorithm against competing state-of-the-art algorithms in terms of quality.
  The performance plot in Figure~\ref{fig:cut_plot} shows the results of our
  experiments performed on \revision{machine A} for all of our benchmark graphs shown in Table~\ref{tab:scalefreegraphstable}.

  Our algorithm gives the best overall quality, usually producing the
  overall best cut. Even in the small fraction of instances where
  other algorithms are best, our algorithm is at most \revision{$7\%$} off.
  The overall solution quality does not heavily depend on the number of PEs used.
  Indeed, more PEs give slightly higher partitioning quality since more initial
  partition attempts are done in parallel. The original fast social configuration of \texttt{KaHIP} as well
  as \texttt{ParHIP} produce worse quality than \texttt{Mt-KaHIP}.
  This is due to the high quality local search scheme that we use; \ie, parallel MLS significantly improves solution quality.
  \texttt{Mt-Metis} with $p = 1$ has worse quality than our algorithm on almost all instances.
  The exceptions are seven mesh type networks and one complex network.
  For \texttt{Mt-Metis} this is expected since it is considerably faster than our algorithm.
  However, \texttt{Mt-Metis} also suffers from deteriorating quality and many
  imbalanced partitions as the number of PEs goes up. This is mostly the case for complex networks.
  This can also be seen from the geometric means of the cut sizes over all instances, including the imbalanced solutions.
  For our algorithm they are \revision{$727.2$K, $713.4$K and $710.8$K for $p = 1, 16, 31$}, respectively.
  For  \texttt{Mt-Metis} they are \revision{$819.8$K, $873.1$K and $874.8$K for $p = 1, 16, 31$}, respectively.
  For \texttt{ParHIP} they are \revision{$809.9$K, $809.4$K and $809.71$K  for $p = 1, 16, 31$}, respectively, and for \texttt{KaHIP} it is $766.2$K.
  For~$p = 31$, the geometric mean cut size of \texttt{Mt-KaHIP}
  is $18.7 \%$ smaller than that of \texttt{Mt-Metis}, $12.2 \%$ smaller than that of
  \texttt{ParHIP} and $7.2 \%$ smaller than that of \texttt{KaHIP}.
  Significance tests that we run indicate that the quality advantage of our solver over the other solvers is statistically significant.
 
  \revision{
    \subsubsection{Effectiveness Tests}
    We now compare the effectiveness of our
    algorithm \texttt{Mt-KaHIP} against our competitors
    using one processor and $31$ processors of machine A.  The idea is to give the
    faster algorithm the same amount of time as the slower algorithm
    for additional repetitions that can lead to improved solutions.%
    \footnote{\revision{Indeed, we asked Dominique LaSalle how
      to improve the quality of Mt-Metis at the expense
      of higher running time and he independently suggested to make
      repeated runs.}}  We have improved an approach used in~\cite{kaffpa}
    to extract more information out of a moderate number
    of measurements. Suppose we want to compare $a$ repetitions of
    algorithm $A$ and $b$ repetitions of algorithm $B$ on instance~$I$.
    We generate a \emph{virtual instance} as follows: We first
    sample one of the repetitions of both algorithms.  Let $t^1_A$ and
    $t^1_B$ refer to the observed running times. Wlog assume $t^1_A\geq t^1_B$.
    Now we sample (without replacement) additional
    repetitions of algorithm $B$ until the total running time
    accumulated for algorithm $B$ exceeds $t^1_A$. Let $t^{\ell}_B$
    denote the running time of the last sample. We accept the last sample of~$B$
    with probability
    $(t^1_A-\sum_{1<i<\ell}t^i_B)/t^{\ell}_B$.
    \begin{theorem}
    The expected total running time of accepted samples for~$B$ is the same as the
    running time for the single repetition of~$A$.
    \end{theorem}
    \begin{proof}
   	Let $t = \sum_{1<i<\ell}t^i_B$.
    Consider a random variable $T$ that is the total time of sampled repetitions.
    With probability $p = \frac{t^1_A - t}{t^{\ell}_B}$, we accept $\ell$-th sample and
    with probability~$1 - p$ we decline it.
    Then
    \begin{equation}
    \begin{split}
    \E[T]   &= p \cdot (t + t^{\ell}_B) + (1 - p) \cdot t \\
           & = \frac{t^1_A - t}{t^{\ell}_B} \cdot (t + t^{\ell}_B) + 
    (1 - \frac{t^1_A - t}{t^{\ell}_B}) \cdot t = t^1_A
    \end{split}
    \end{equation}
    \end{proof}	
    
    The quality assumed for~$A$ in this virtual instance is
    the quality of the only run of algorithm~$A$. The quality assumed for~$B$ is
    the best quality obverved for an accepted sample for~$B$.

  For our effectiveness evaluation, we used $20$ virtual instances for
  each pair of graph and $k$ derived from 10 repetitions of each
  algorithm.  Figure~\ref{fig:cut_pairwise_plot} presents the
  performance plot for \texttt{Mt-KaHIP} and \texttt{Mt-Metis}
  for different number of processors.
  As we can see, even with additional running time \texttt{Mt-Metis} has
  mostly worse quality than \texttt{Mt-KaHIP}.
  Consider the effectiveness test where \texttt{Mt-KaHIP} and \texttt{Mt-Metis} run with $31$ threads.
   In $80.4 \%$ of the virtual instances \texttt{Mt-KaHIP} has better quality than \texttt{Mt-Metis}. 
   In the worst-case, \texttt{Mt-KaHIP} has only a $5.5\%$ larger cut than \texttt{Mt-Metis}.
  
  Figure~\ref{fig:cut_pairwise_plot_1} presents the
  performance plot for \texttt{Mt-KaHIP} and \texttt{ParHIP}
  for different number of processors.
  As we can see, even with additional running time \texttt{ParHIP} has
  mostly worse quality than \texttt{Mt-KaHIP}.
  Consider the effectiveness test where \texttt{Mt-KaHIP} and \texttt{ParHIP} run with $31$ threads.
  In $96.5 \%$ of the virtual instances, \texttt{Mt-KaHIP} has better quality than \texttt{ParHIP}.
  In the worst-case, \texttt{Mt-KaHIP} has only a $5.4\%$ larger cut than \texttt{ParHIP}.
  
  Figure~\ref{fig:cut_pairwise_plot_2} presents the
  performance plot for \texttt{Mt-KaHIP} and \texttt{KaHIP}
  for different number of processors.
  As we can see, even with additional running time \texttt{KaHIP} has
  mostly worse quality than \texttt{Mt-KaHIP}.
  Consider the effectiveness test where \texttt{Mt-KaHIP} runs with $31$ threads.
  In $98.9 \%$ of the virtual instances, \texttt{Mt-KaHIP} has better quality than \texttt{KaHIP}.
  In the worst-case, \texttt{Mt-KaHIP} has only a $3.5\%$ larger cut than \texttt{KaHIP}.
  
  }
  \begin{figure}[ht!]
    \includegraphics[width=.475\textwidth]{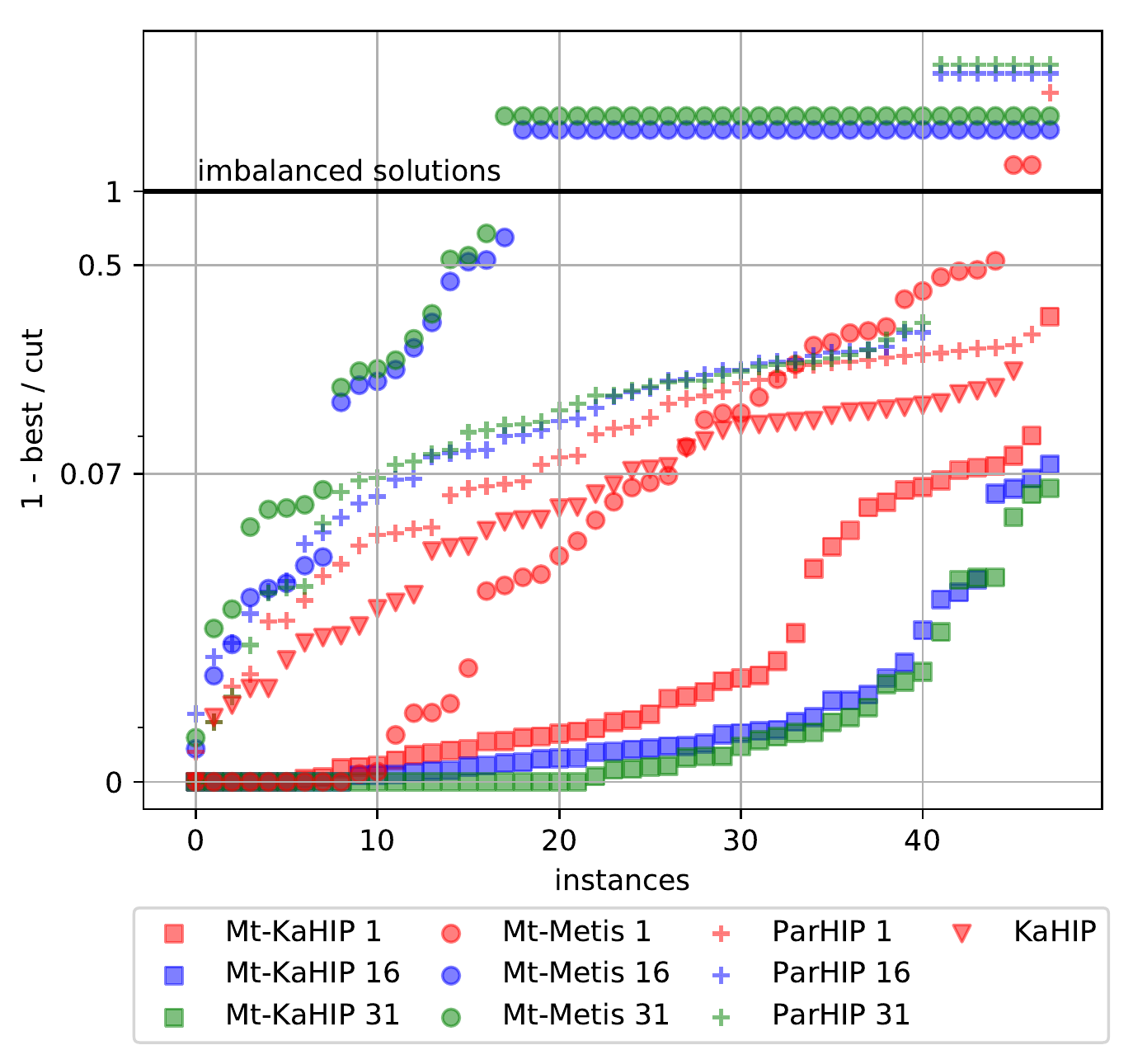}
    \caption{Performance plot for the cut size. The number behind the algorithm name denotes the number of threads used.
    	}
    \label{fig:cut_plot}
  \end{figure}
  \begin{figure}[ht!]
  	\includegraphics[width=.475\textwidth]{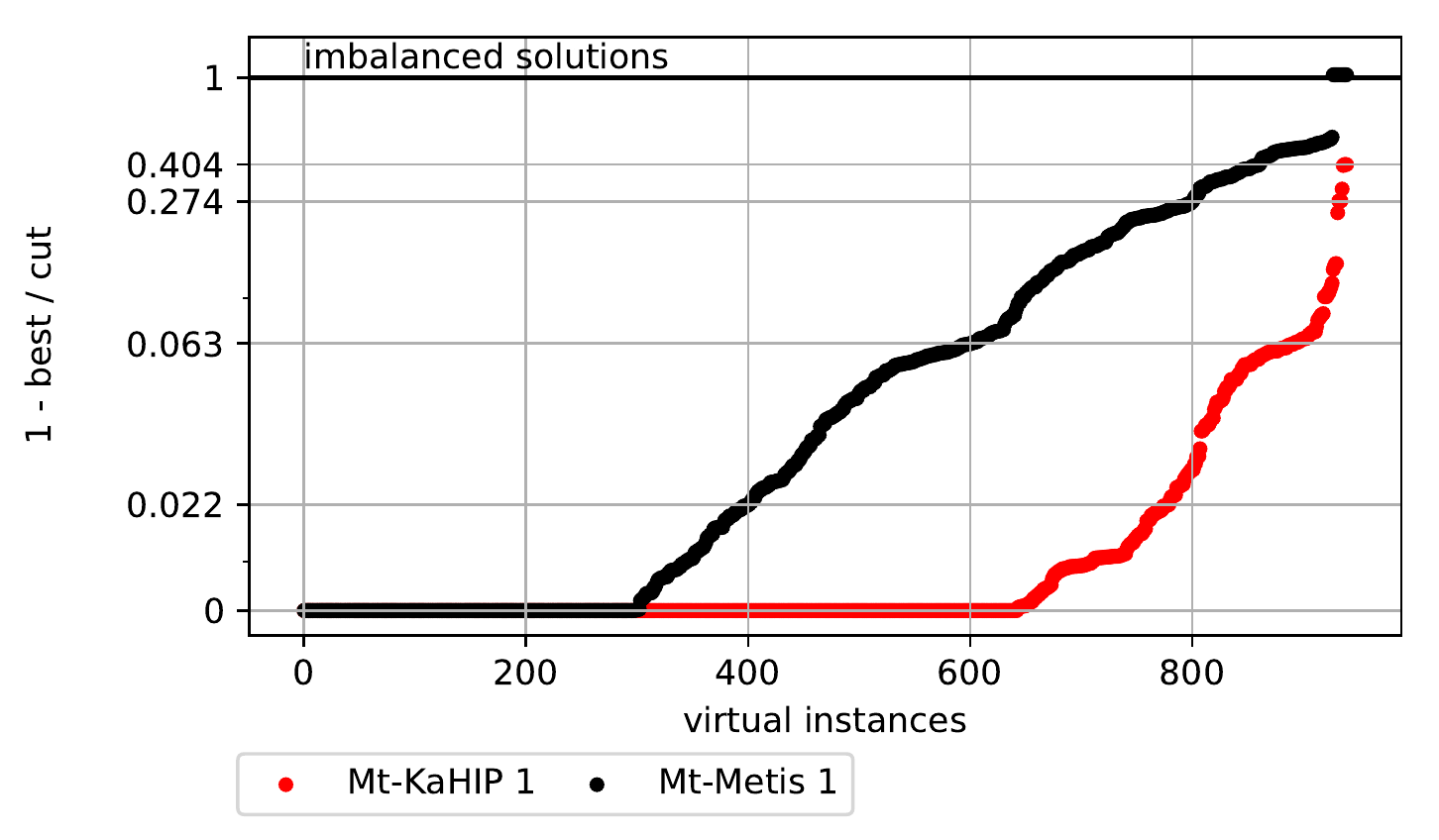}
  	\includegraphics[width=.475\textwidth]{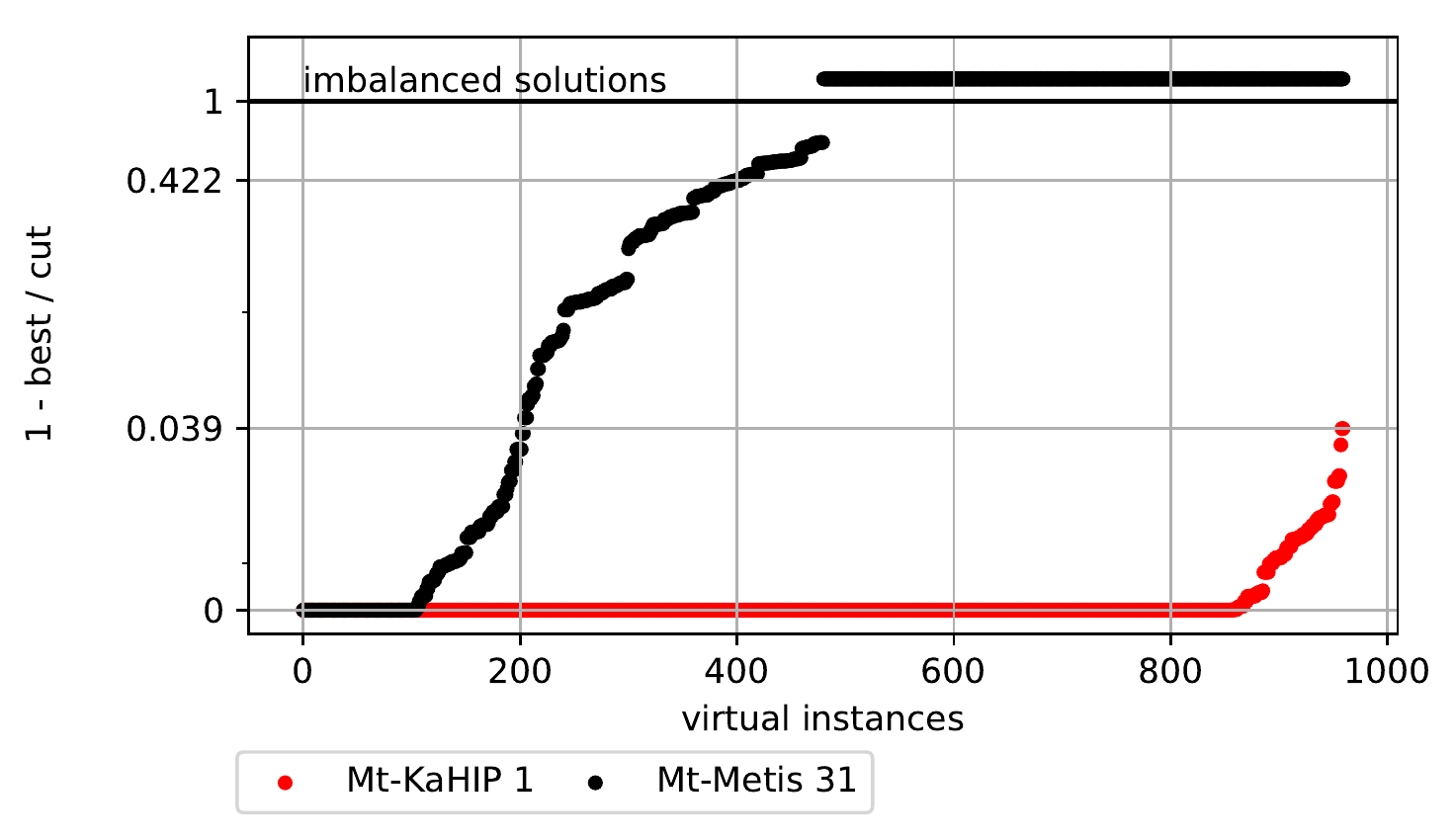}
  	\includegraphics[width=.475\textwidth]{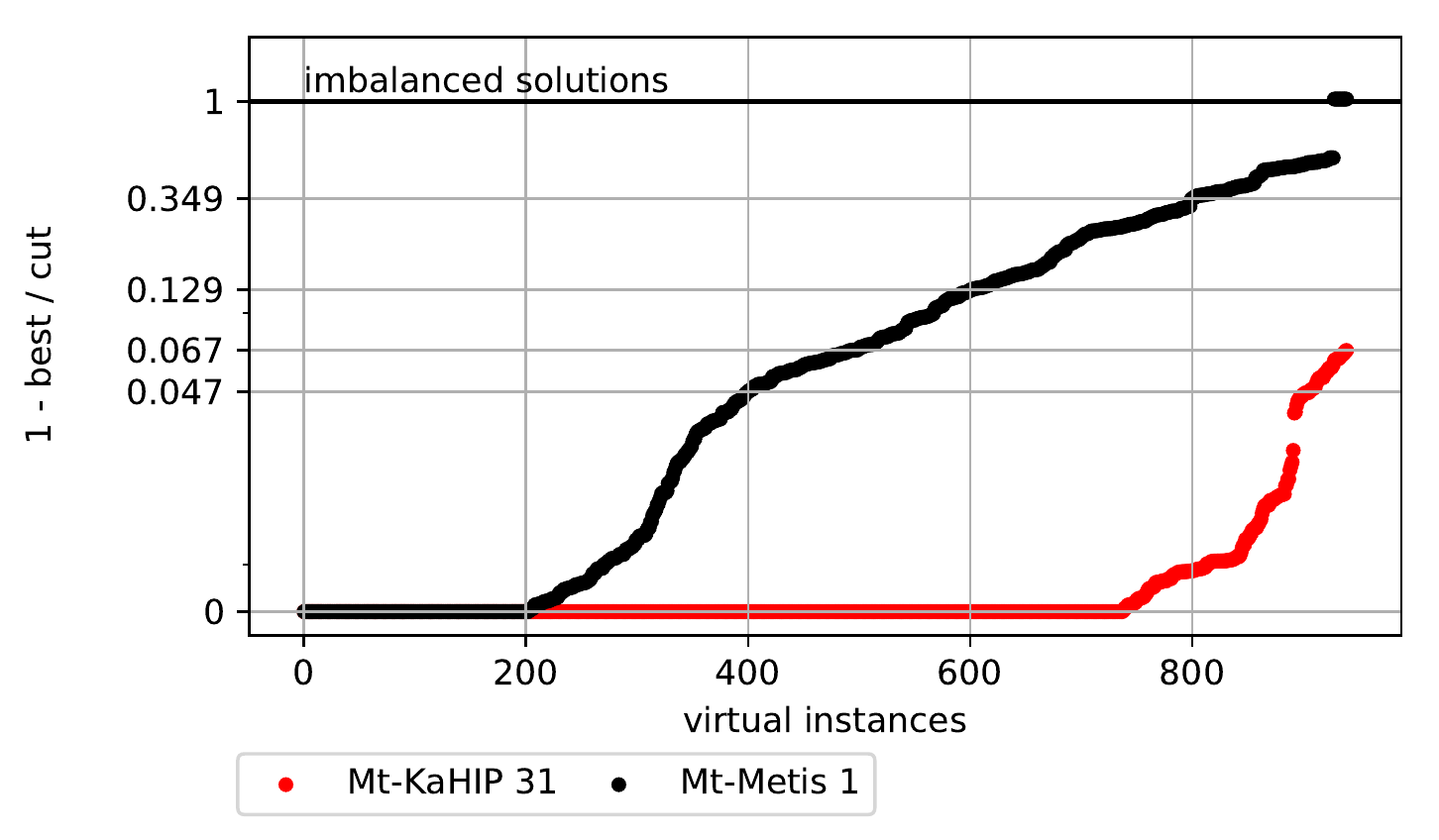}
  	\includegraphics[width=.475\textwidth]{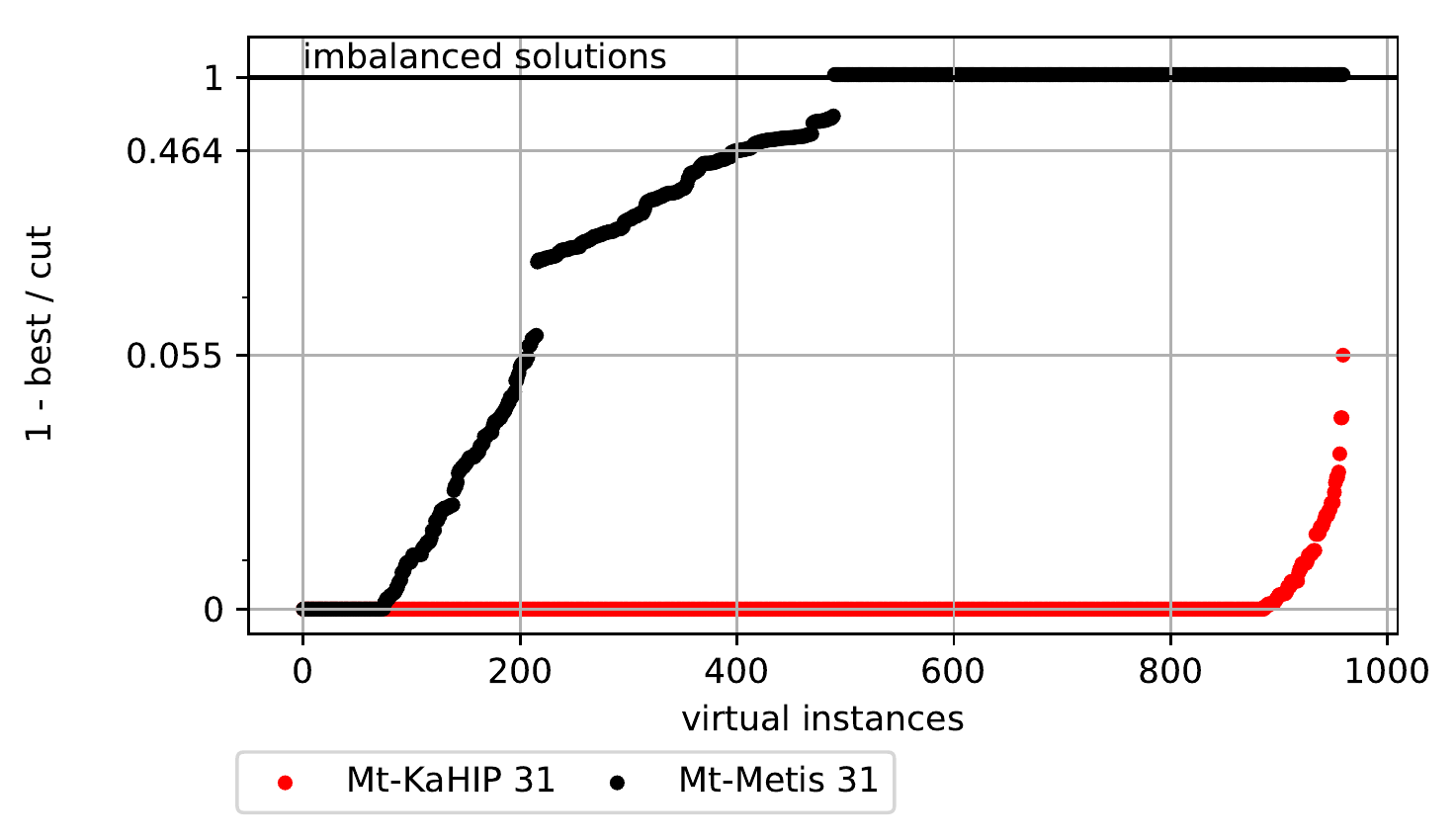}
  	\caption{Effectiveness tests for \texttt{Mt-KaHIP} and \texttt{Mt-Metis}. The number behind the algorithm name denotes the number of threads used.
  	}
  	\label{fig:cut_pairwise_plot}
  \end{figure}
  \begin{figure}[ht!]
  	\includegraphics[width=.475\textwidth]{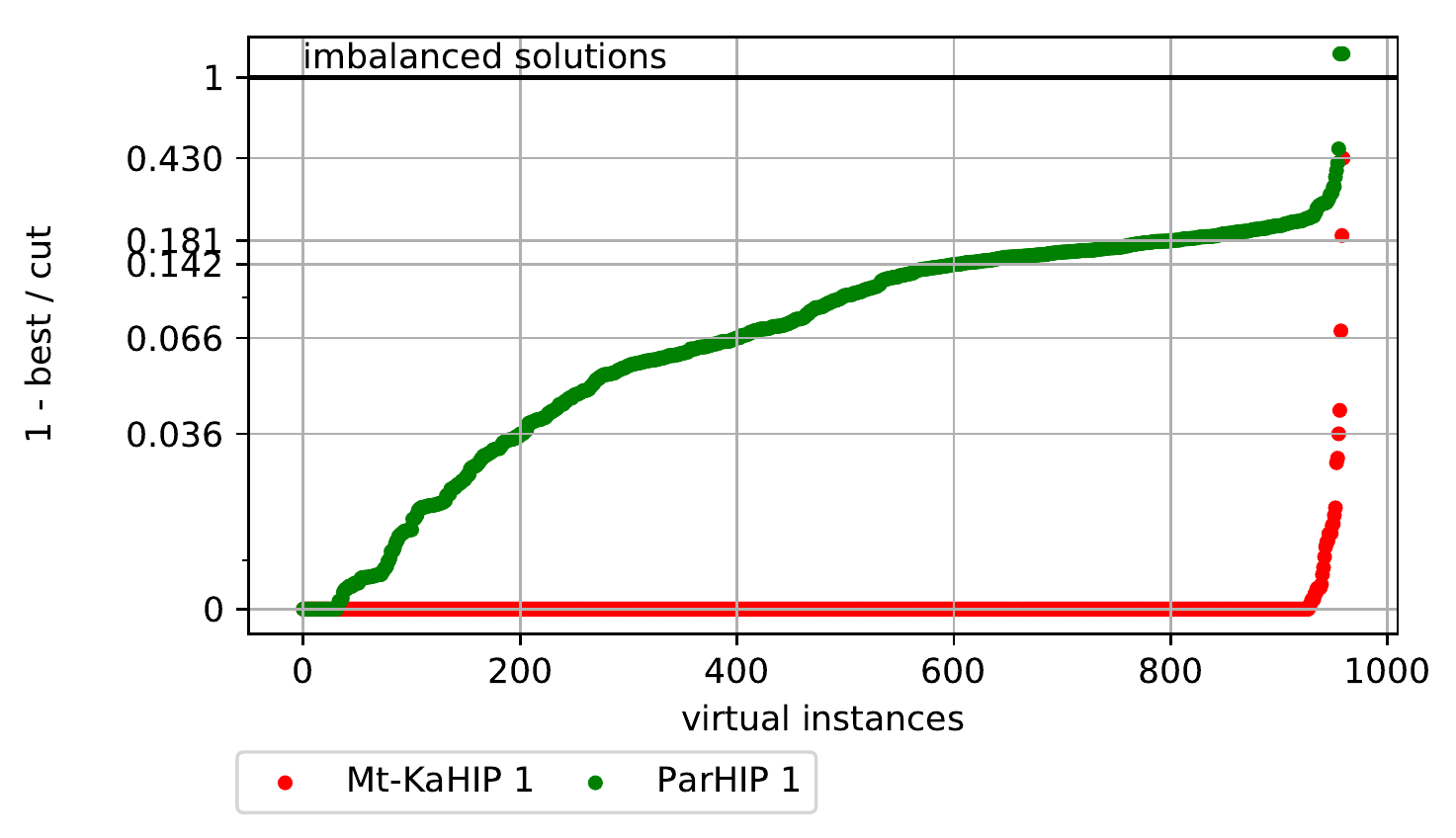}
  	\includegraphics[width=.475\textwidth]{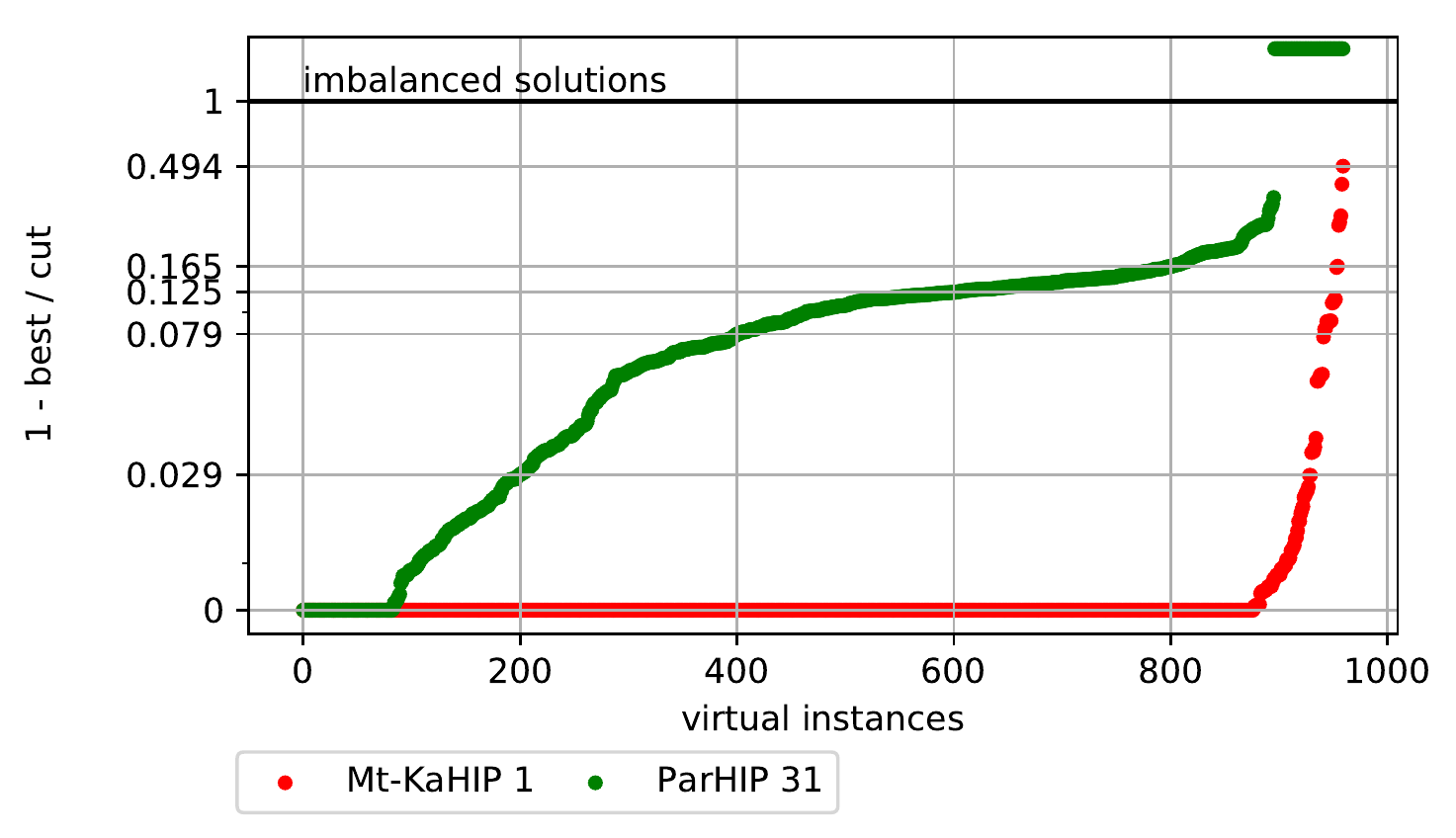}
  	\includegraphics[width=.475\textwidth]{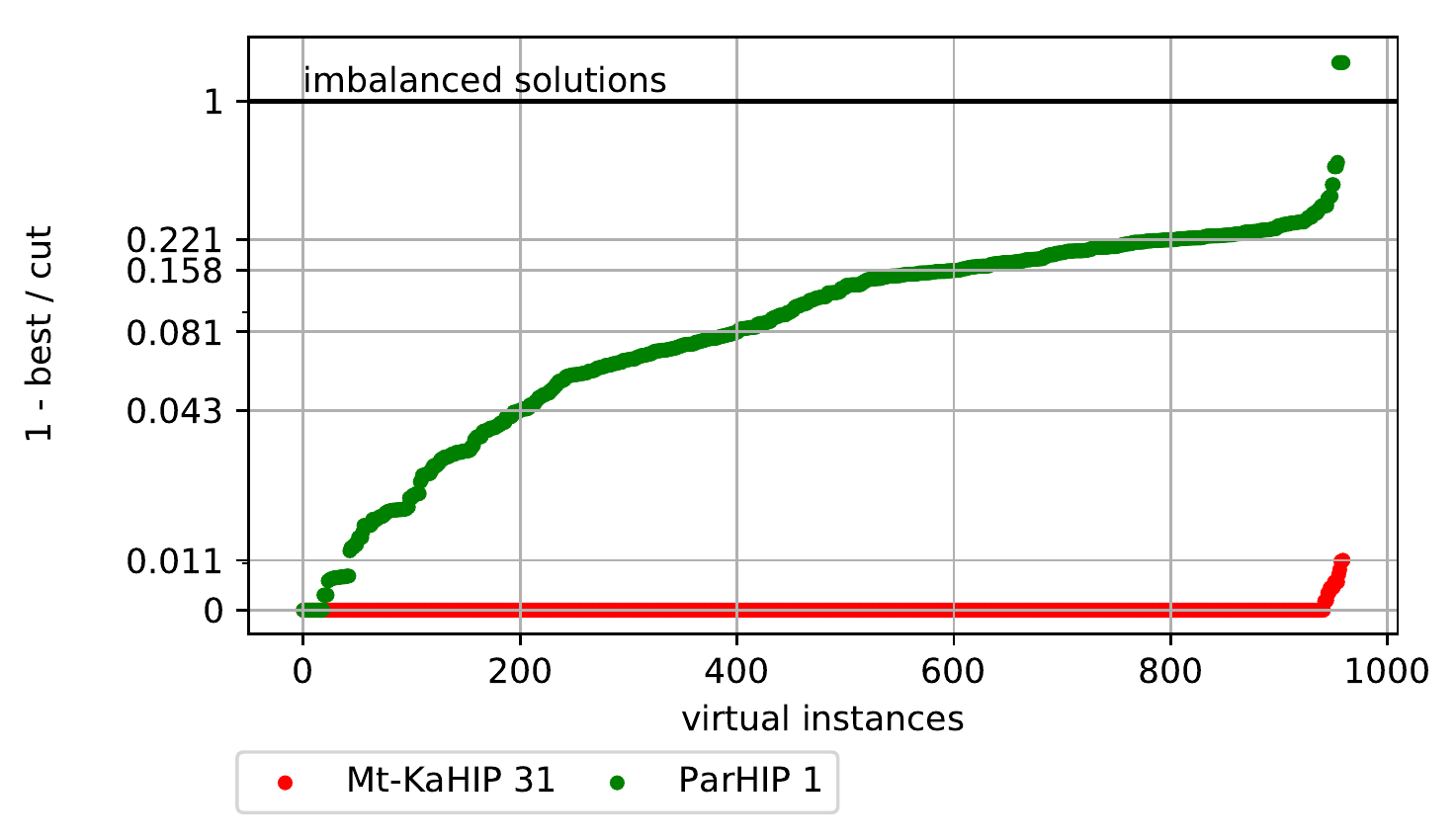}
  	\includegraphics[width=.475\textwidth]{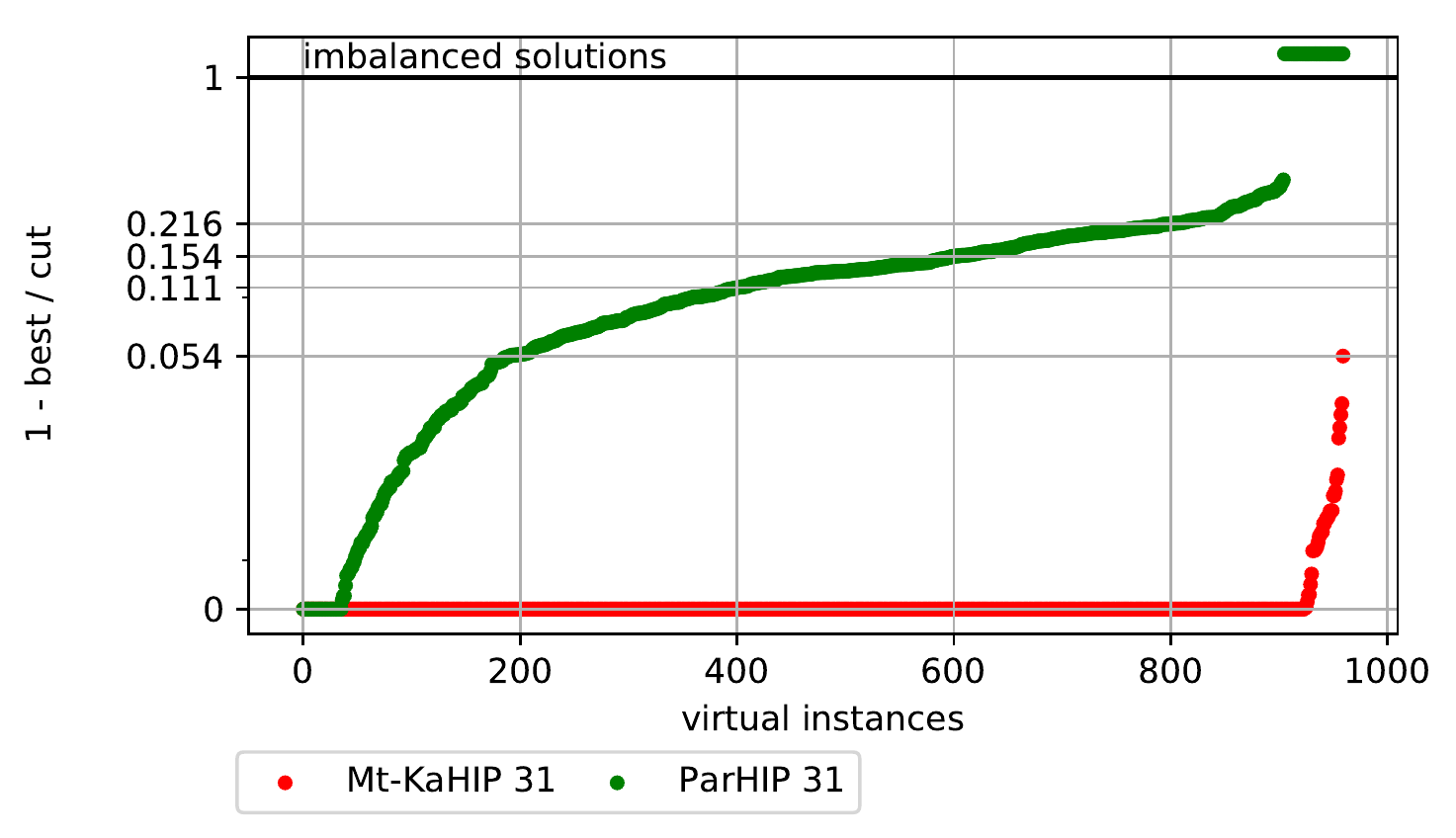}
  	\caption{Effectiveness tests for \texttt{Mt-KaHIP} and \texttt{ParHIP}. The number behind the algorithm name denotes the number of threads used.
  	}
  	\label{fig:cut_pairwise_plot_1}
  \end{figure}
  \begin{figure}[ht!]
  	\includegraphics[width=.475\textwidth]{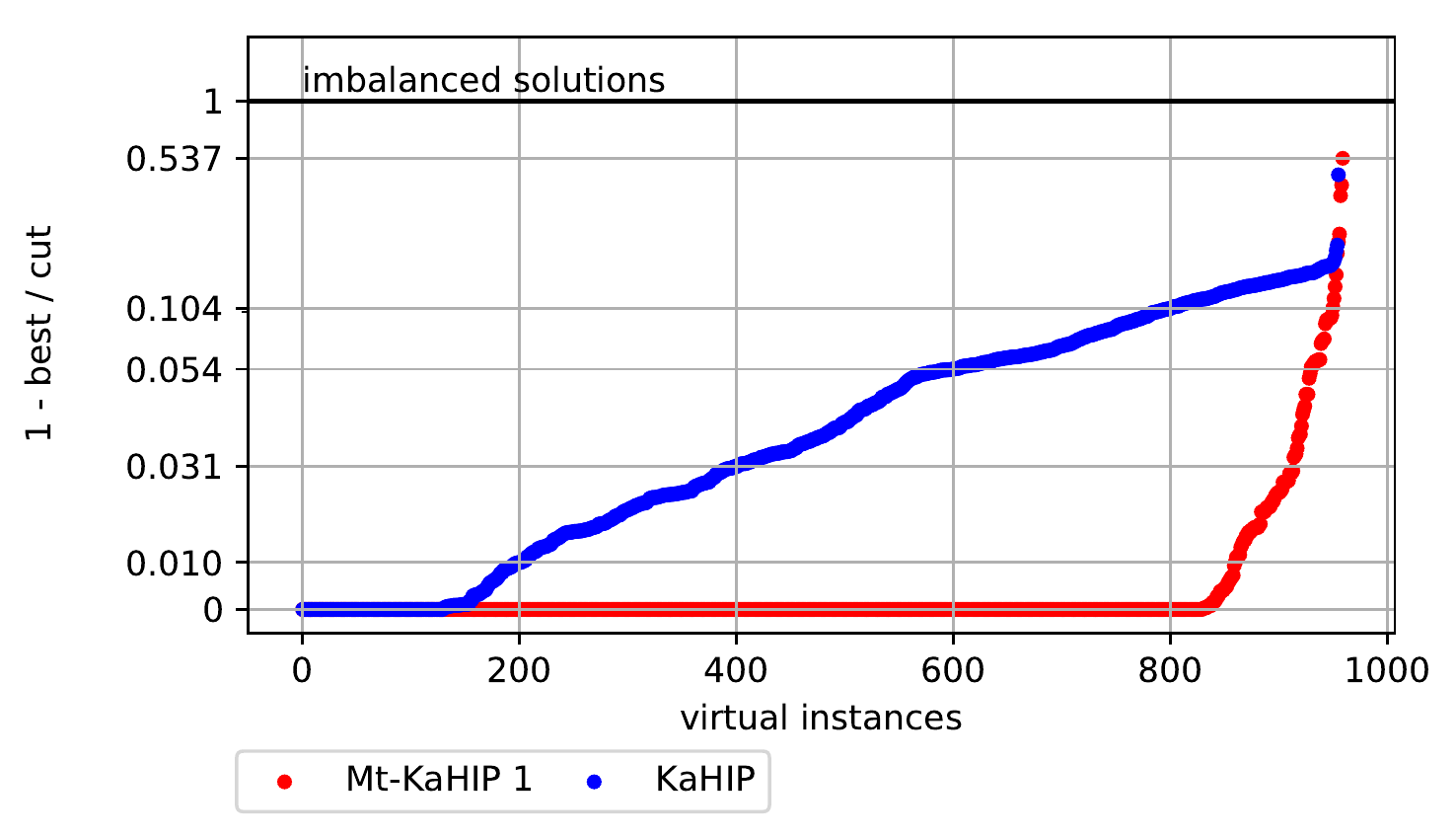}
  	\includegraphics[width=.475\textwidth]{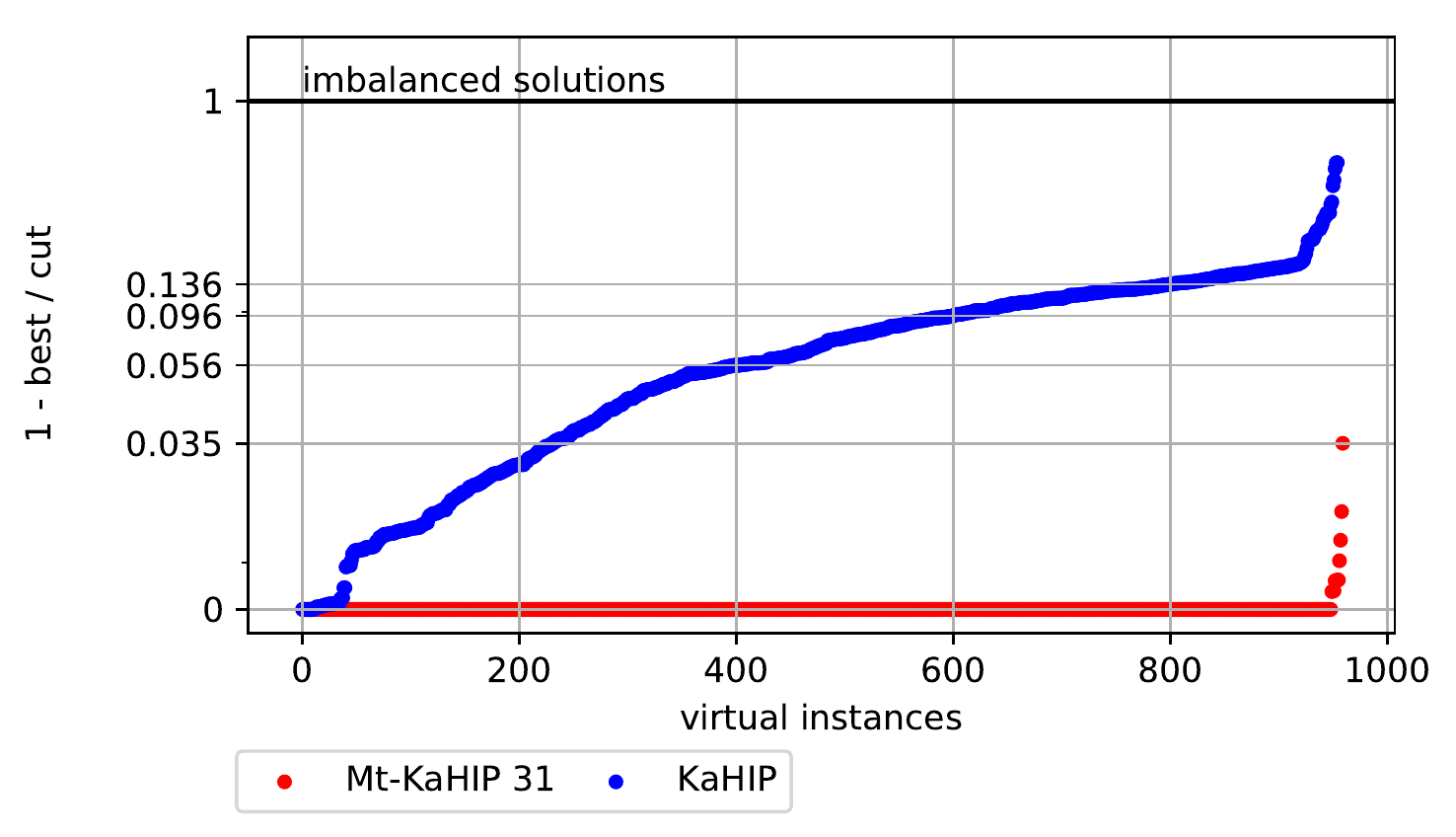}
  	\caption{Effectiveness tests for \texttt{Mt-KaHIP} and \texttt{KaHIP}. The number behind the algorithm name denotes the number of threads used.
  	}
  	\label{fig:cut_pairwise_plot_2}
  \end{figure}
  \subsection{Speed-up and Running Time Comparison}
    In this section, we compare the speed-ups and the running times of our algorithm against competing algorithms.
    We calculate a relative speed-up of an algorithm as a ratio between its running time
    (averaged over ten repetitions) and its running time with $p = 1$. 
    Figure~\ref{fig:ls_rt_16} 
    show scatter plots with speed-ups and time per edge for a full algorithm execution and local search
    (for our algorithm it is MLS) on \revision{machine A}.
    Additionally, we calculate the harmonic mean \textit{only} for instances that were partitioned in ten repetitions
    \textit{without imbalance}.
    Note that among  \revision{the top $20$ speed-ups of $\texttt{Mt-Metis}$ $60\%$ correspond to imbalanced
    instances (\textit{Mt-Metis 31 imbalanced})} thus we believe it is fair to exclude them.

    The harmonic mean full speed-up of our algorithm, $\texttt{Mt-Metis}$ and $\texttt{ParHIP}$ for \revision{$p = 31$
    are $9.1, 11.1$ and $9.5$}, respectively.
    The harmonic mean local search speed-up of our algorithm, $\texttt{Mt-Metis}$ and $\texttt{ParHIP}$
    are  \revision{$13.5, 6.7$ and $7.5$}, respectively.
    Our full speed-ups are comparable to that of $\texttt{Mt-Metis}$
    but our local search speed-ups are significantly better than that of $\texttt{Mt-Metis}$.
    The geometric mean full time per edge of our algorithm, $\texttt{Mt-Metis}$ and $\texttt{ParHIP}$
    are  \revision{$52.3$ nanoseconds (ns), $12.4$ [ns] and $121.9$ [ns]}, respectively.
    The geometric mean local search time per edge of our algorithm,
    $\texttt{Mt-Metis}$ and $\texttt{ParHIP}$ are  \revision{$3.5$ [ns], $2.1$ [ns] and $16.8$~[ns]}, respectively.
    Note that with increasing number of edges, our algorithm has comparable time per edge to 
    $\texttt{Mt-Metis}$.
    Superior speed-ups of parallel MLS are due to minimized interactions between PEs and using cache-aware hash tables locally.
    Although on average, our algorithm is slower than \texttt{Mt-Metis}, we consider this as a fair trade off
    between the quality and the running time.
    We also dominate \texttt{ParHIP} in terms of quality and running times.
    




\begin{figure}[h!]
	\includegraphics[width=.43\textwidth]{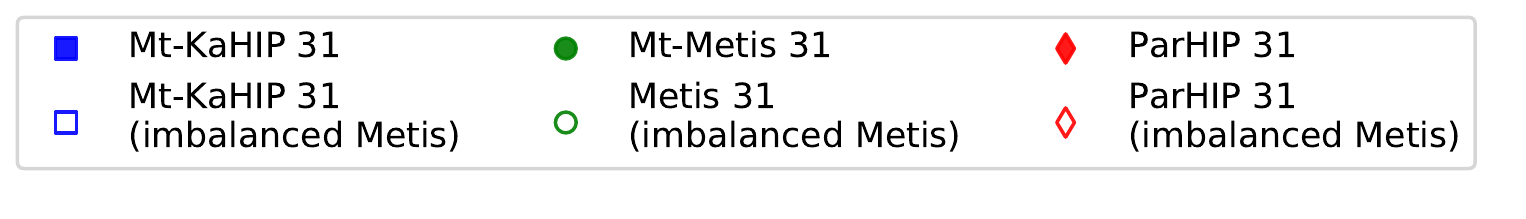}
	\includegraphics[width=.43\textwidth]{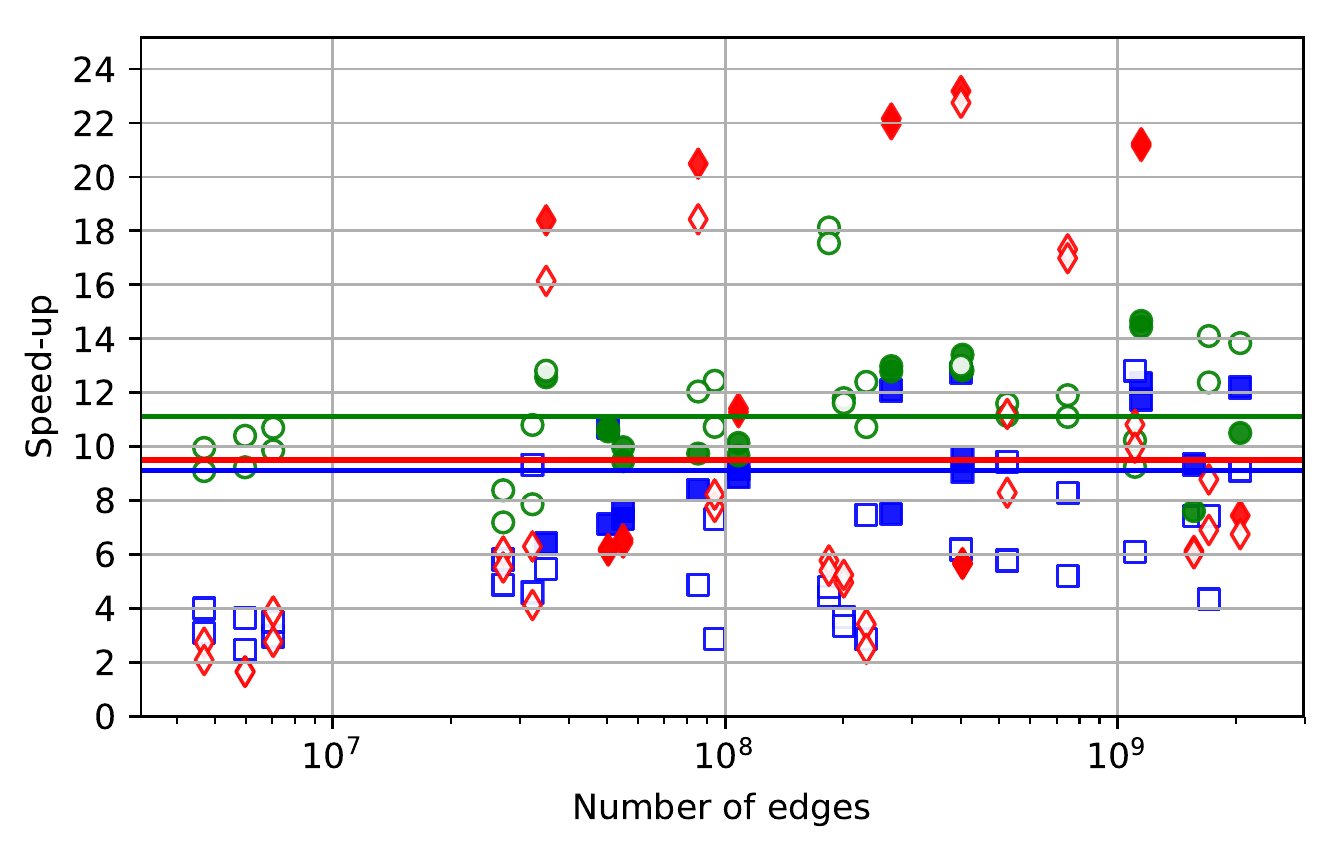}
	\label{fig:full_sp_16}

	\includegraphics[width=.43\textwidth]{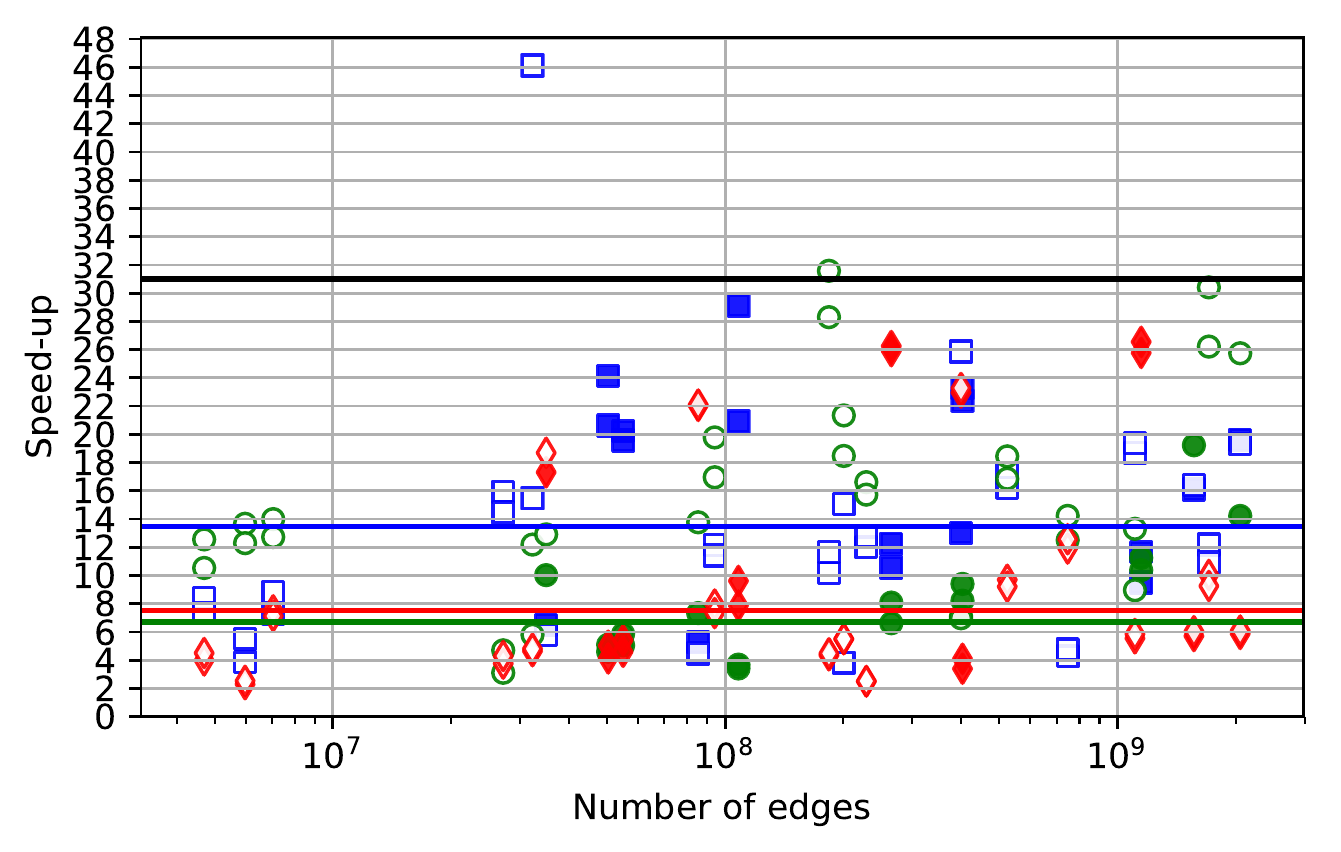}
	\label{fig:ls_sp_16}

	\includegraphics[width=.43\textwidth]{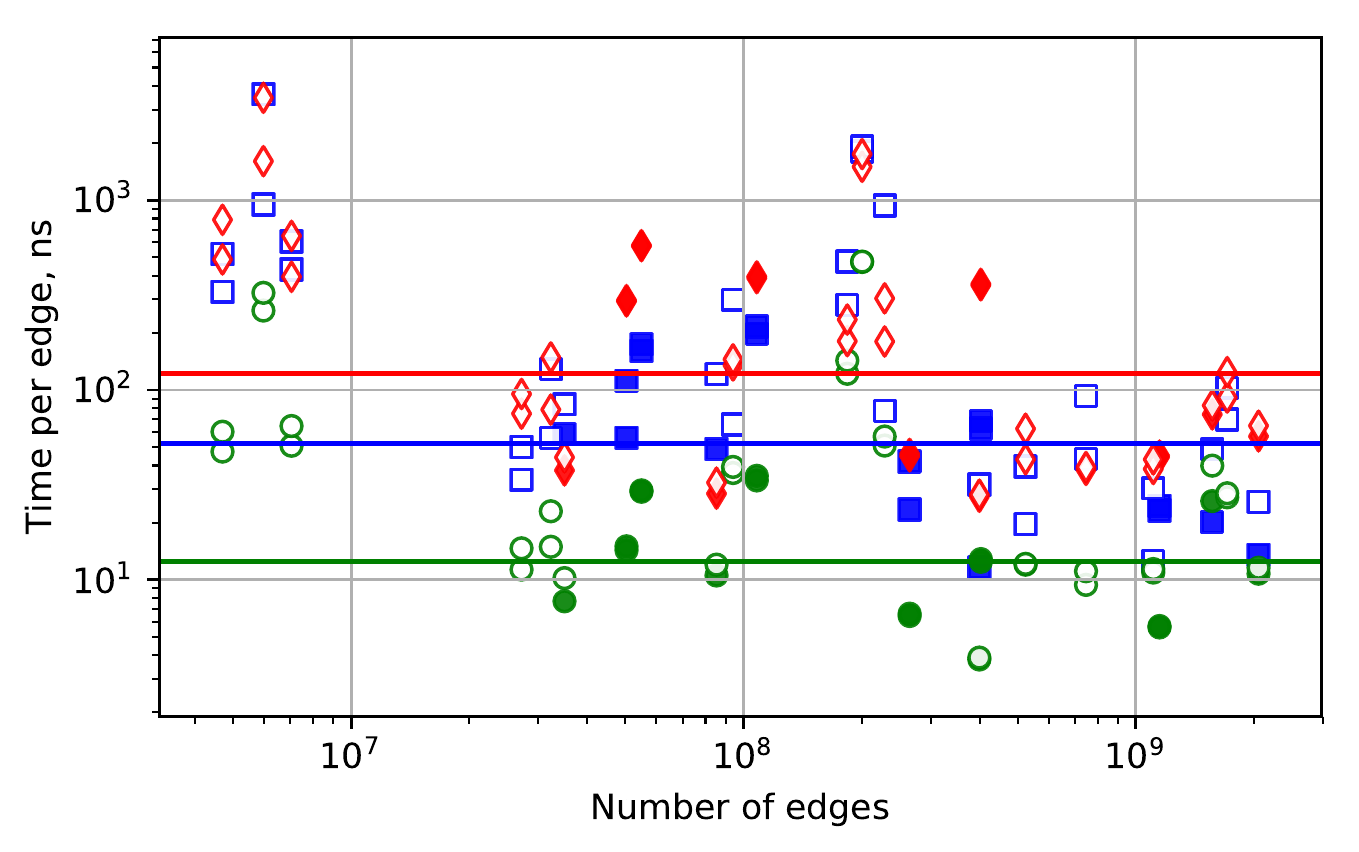}
	\label{fig:full_rt_16}

	\includegraphics[width=.43\textwidth]{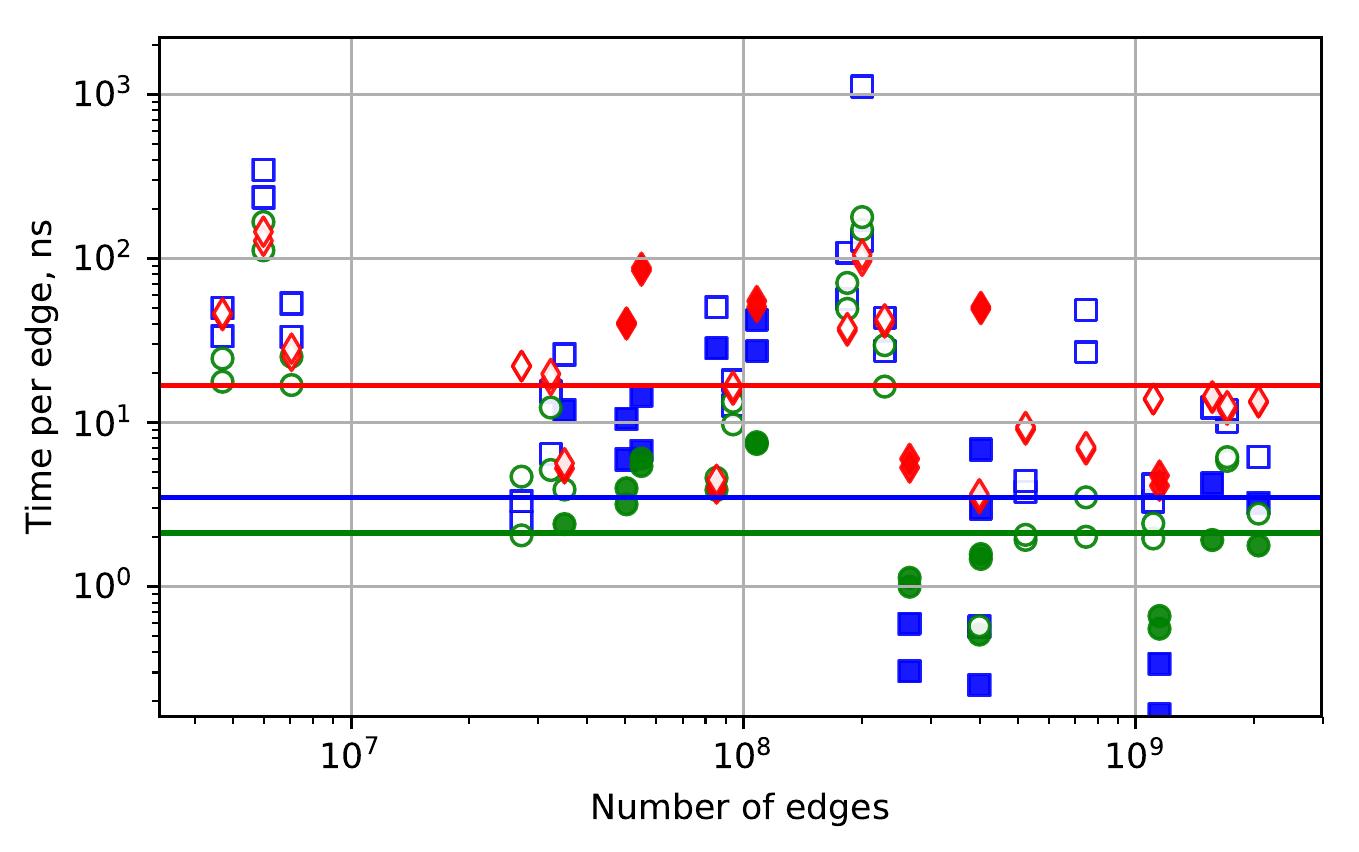}
        \caption{From top to bottom for $p=31$: legend, a) full speed-up,  b) local search speed-up, c) full running time per edge in nanoseconds, d) local search running time per edge in nanoseconds. Non-black horizontal lines are harmonic and geometric means.}
	\label{fig:ls_rt_16}
\end{figure}

    %

\revision{
	\subsection{Influence of Algorithmic Components}
	We now analyze how the parallelization of the different components affects solution quality of
	partitioning and present the speed-ups of each phase.
	We perform experiments on machine $B$ with configurations of our algorithm in which only one of the components (coarsening, initial partitioning, uncoarsening) is parallelized.
	The respective parallelized component of the algorithm uses $16$ processors and the other components run sequentially.
	Running the algorithm with parallel coarsening decreases the geometric mean of the cut by 0.7\%, with parallel initial partitioning decreases the cut by 2.3\% and with parallel local search decreases the cut by 0.02\%. 
	Compared to the full sequential algorithm, we conclude that running the algorithm with any parallel component either does not affect solution quality or improves the cut slightly on average.
	The parallelization of initial partitioning gives better cuts since it computes more initial partitions than the sequential~version.
	
	To show that the parallelization of each phase is important, we consider instances
	where one of the phases runs significantly longer than other phases.
	To do so, we perform experiments on machine $A$ using $p=31$. 
	For the graph rgg26 and $k = 16$, the \emph{coarsening phase} takes $91\%$ of the running time and
	its parallelization gives a speed-up of $13.6$ for $31$ threads and a full
	speed-up of $12.4$.
	For the graph webbase-2001 and $k = 16$, the \emph{initial partitioning phase} takes $40\%$
	of the running time and its parallelization gives a speed-up of $6.1$ and the overall
	speed-up is $7.4$.
	For the graph it-2004 and $k = 64$, the \emph{uncoarsening phase} takes $57\%$ of the runnung time and its parallelization gives a speed-up of $13.0$ and the overall
	speed-up is $9.1$.
	The harmonic mean speed-ups of the coarsening phase, the initial partitioning phase and
	the uncoarsening phase for $p = 31$ are $10.6$, $2.0$ and $8.6$,~respectively.
	
}

\revision{
\subsection{Memory consumption}
We now look at the memory consumption of our parallel algorithm on the three biggest graphs
of our benchmark set uk-2005, webbase-2001 and it-2004 for $k = 16$ (for $k = 64$ they are comparable).
The amount of memory needed by our algorithm for these graphs is $26.1$GB, $33.7$GB, and $34.5$GB for $p = 1$ on machine A, respectively.
For $p = 31$, our algorithm needs $30.5$GB, $45.3$GB, and $38.3$GB.
We observe only small memory overheads when increasing the number of processors.
We explain these by the fact that all data structures created by each processor are
either of small memory size (copy of a coarsened graph) or the data is distributed between them
approximately uniformly (a hash table in Multi-try $k$-way Local Search).
The amount of memory needed by \texttt{Mt-Metis} for these graphs is $46.8$GB, $53.3$GB, and $61.9$GB for $p = 1$, respectively.
For $p = 31$, \texttt{Mt-Metis} needs $59.4$GB, $67.7$GB, and $68.7$GB.
Summarizing, our algorithm consumes $48.7 \%$, $33.1 \%$, $44.3 \%$ less memory for these graphs
for $p = 31$.
Although, both algorithms have relatively little memory overhead for parallelization.}

\section{Conclusion and Future Work}
\label{sec:conclusion}

Graph partitioning is a key prerequisite for efficient large-scale parallel graph algorithms. 
We presented an approach to multi-level shared-memory parallel graph partitioning that guarantees balanced solutions, 
shows high speed-ups for a variety of large graphs and yields very good quality independently of the number of cores 
used.
Previous approaches have problems with recently grown structural complexity of networks that need partitioning -- 
they often show a negative trade-off between speed and quality.
Important ingredients of our algorithm include parallel label propagation for both coarsening and refinement,
parallel initial partitioning, a simple yet effective approach to parallel localized local search, and fast locality 
preserving hash tables.
Considering the good results of our algorithm, we want to further improve it and release its implementation.
More precisely, we are planning to further improve scalability of parallel coarsening and parallel MLS.
An interesting problem is how to apply moves in Section~\ref{parallel_mgp:refinement}
without the gain recalculation.
The solution of this problem will increase the performance of parallel MLS.
Additionally, we are planning to integrate a high quality parallel matching algorithm for the coarsening phase that
allows to receive better quality for mesh-like graphs. Further quality improvements should be possible by integrating a parallel version of the flow based techniques used in KaHIP.


\section*{Acknowledgment}
We thank Dominique LaSalle for helpful discussions about his hill-climbing local search technique.
The research leading to these results has received funding from the European Research Council under the European Union's Seventh Framework Programme (FP/2007-2013) / ERC Grant Agreement no. 340506.
\bibliographystyle{amsplain}
\bibliography{diss,bibliography,phdthesiscs}
\end{document}